\documentclass[acmsmall]{acmart}
\AtBeginDocument{%
  }

\setcopyright{acmlicensed}
\acmJournal{PACMMOD}
\acmYear{2025} \acmVolume{3} \acmNumber{6 (SIGMOD)}
\acmArticle{345} \acmMonth{12} \acmPrice{}
\acmDOI{10.1145/3769810}

\usepackage{enumitem}
\usepackage{float}
\usepackage{subcaption}
\usepackage{tabularx}
\usepackage[linesnumbered,ruled,noend]{algorithm2e}
\usepackage{blindtext}
\usepackage{bm}
\usepackage{booktabs}
\usepackage{balance}

\newcommand{\secref}[1]{Sec. \ref{#1}}
\newcommand{\figref}[1]{Fig. \ref{#1}}
\newcommand{\tabref}[1]{Tab. \ref{#1}}
\newcommand{\eqnref}[1]{Eqn. \ref{#1}}

\newcommand{\algref}[1]{Alg. \ref{#1}}
\newcommand{\thmref}[1]{Theorem \ref{#1}}

\newcommand{\defref}[1]{Def. \ref{#1}}
\newcommand{\appref}[1]{Appendix \ref{#1}}

\newcommand{\nop}[1]{}
\newcommand{\nosection}[1]{\noindent\textbf{\emph{#1}}}

\usepackage{xspace}
\def\method{\texttt{NeuSO}\xspace}
\def\plainmethod{NeuSO\xspace}

\SetCommentSty{mycommfont}

\newbool{fullversion}
\booltrue{fullversion}   

\begin{document}

\title{\plainmethod: Neural Optimizer for Subgraph Queries}


\author{Linglin Yang}
\email{linglinyang@stu.pku.edu.cn}
\orcid{0000-0001-9480-6088}
\affiliation{%
  \institution{Peking University}
  \city{Beijing}
  \country{China}
}

\author{Lei Zou}
\email{zoulei@pku.edu.cn}
\orcid{0000-0002-8586-4400}
\affiliation{%
  \institution{Peking University}
  \city{Beijing}
  \country{China}
}

\author{Chunshan Zhao}
\email{zhaochunshan@pku.edu.cn}
\orcid{0009-0002-2082-1669}
\affiliation{
  \institution{Peking University}
  \city{Beijing}
  \country{China}
}


\begin{abstract}

Subgraph query is a critical task in graph analysis with a wide range of applications across various domains. Most existing methods rely on heuristic vertex matching orderings, which may significantly degrade enumeration performance for certain queries. While learning-based optimizers have recently gained attention in the context of relational databases, they cannot be directly applied to subgraph queries due to the heterogeneous and schema-flexible nature of graph data, as well as the large number of joins involved in subgraph queries.
These complexities often leads to inefficient online performance, making such approaches impractical for real-world graph database systems. 
To address this challenge, we propose \method, a novel learning-based optimizer for subgraph queries that achieves both high accuracy and efficiency. 
\method features an efficient query graph encoder and an estimator which are trained using a multi-task framework to estimate both subquery cardinality and execution cost.
Based on these estimates, \method employs a top-down plan enumerator to generate high-quality execution plans for subgraph queries.
Extensive experiments on multiple datasets demonstrate that \method outperforms existing subgraph query ordering approaches in both performance and efficiency.
\end{abstract}

\ccsdesc[500]{Information systems~Query optimization}
\ccsdesc[500]{Information systems~Query planning}

\keywords{Subgraph Query, Query Optimization}

\received{April 2025}
\received[revised]{July 2025}
\received[accepted]{August 2025}
\maketitle

\section{Introduction}


Graph models have gained considerable attention in recent years due to their straightforward and clear representation of relationships between entities. They are widely applied across various fields, including social networks \cite{fan2012graph}, biology \cite{bonnici2013subgraph}, and knowledge graphs \cite{kim2015taming, perez2009semantics}. To retrieve a user-interested subgraph from the original big data graph, users may invoke subgraph queries on the data graph, which is a fundamental type of graph queries in the field of graph databases \cite{angles2008survey, barcelo2013querying, zou2011gstore}.

Specifically, given a query graph and a data graph, subgraph query is to find all matches of the query graph in the data graph, which preserves the label constraints of vertices and topology constraints of edges. Typically, the processing of a subgraph query involves three stages: filtering, planning, and enumeration. The filtering stage employs various techniques to extract a subset of vertices in the data graph (i.e., potential matching candidates) for each query vertex, thereby narrowing the search space. Next, in the planning stage, a matching order for the query vertices is generated. Finally, the matching results are obtained by enumerating according to the specified matching order.

The matching order plays a significant role in execution efficiency. The choice of matching order can result in performance variations spanning several orders of magnitude \cite{leis2015good, sun2020memory, zhang2024comprehensive}. While much research has focused on improving the speed of subgraph matching, particularly through various filtering techniques \cite{bi2016efficient,bhattarai2019ceci,han2019efficient,arai2023gup} and reducing duplicate computations during enumeration \cite{bi2016efficient,han2019efficient,jin2023circinus}, few works have addressed the matching order selection problem effectively. Most existing subgraph matching approaches either depend on heuristic strategies or leverage dynamic programming (DP) techniques to determine the matching order. Nevertheless, these methods often face several critical limitations:

\begin{itemize}[leftmargin=8pt]
    \item \textit{Structured-based heuristic methods} \cite{bonnici2013subgraph,sun2020rapidmatch} are based solely on the structure of the query graph, resulting in the same matching order for different data graphs if the query graph is the same. This can lead to significant performance degradation when applied to biased data graphs. 
    \item \textit{Rule-based heuristic methods} \cite{he2008graphs,bi2016efficient,bhattarai2019ceci,han2019efficient,kim2021versatile} rely on simple rules, such as prioritizing vertices with large degrees or small candidate sizes, without considering more detailed factors. This oversight can easily lead to suboptimal matching orders, which may be several orders of magnitude worse than the optimal solution.
    \item \textit{Dynamic programming-based methods} \cite{yang2022gcbo, mhedhbi2019optimizing, feng2023kuzu, neo4j} are typically limited to small-scale query graphs and can become computationally expensive or even intractable for large query graphs.
\end{itemize}

\subsection{Motivation}
Traditional optimizers for relational databases are based on cost estimation and follow a pipeline structure comprising three main components: the cardinality estimator, the cost estimator, and the plan enumerator, as shown in \figref{fig:method_classification}. When a query is received, the plan enumerator generates possible execution plans, then selects the plan with the lowest estimated execution cost for actual execution. One of the most critical factors influencing execution cost is cardinality, which refers to the size of the intermediate results. As a result, the cost estimator often relies on the cardinality estimator to estimate the execution cost for each candidate plan.

With the rise of machine learning, a number of learning-based optimization methods for relational databases have emerged. These methods can be categorized into the following three approaches, as shown in \figref{fig:method_classification}:

\begin{figure}
    \setlength{\abovecaptionskip}{0.13cm}
    \setlength{\belowcaptionskip}{-0.3cm}
    \centering
    \includegraphics[width=0.7\linewidth]{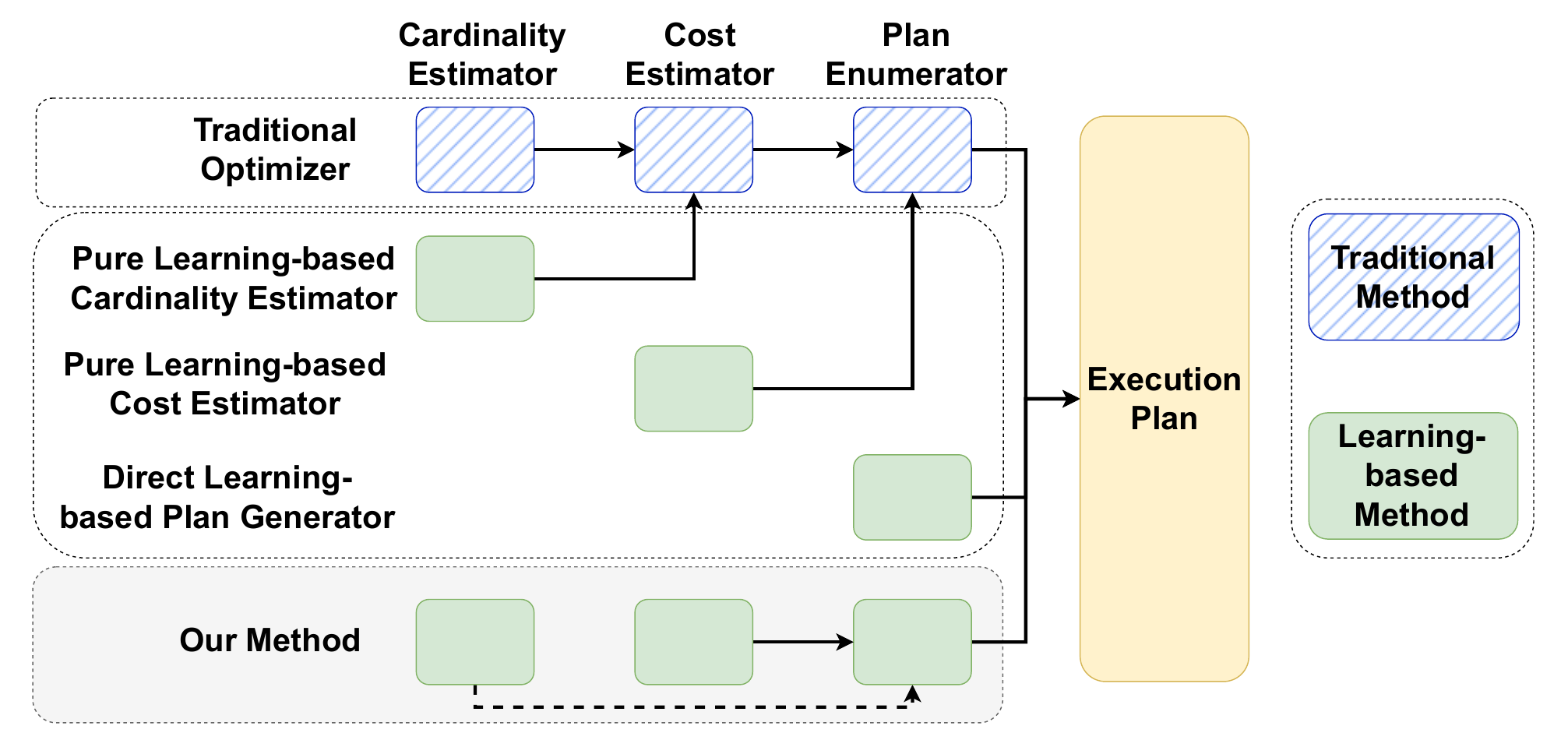}
    \caption{Optimizer classification. The blue block represents traditional method, while the green block represent learning-based method.}
    \label{fig:method_classification}
\end{figure}

\begin{itemize}[leftmargin=8pt]
    \item \textbf{Pure learning-based cardinality estimator} \cite{kipf2019mscn,hilprecht13deepdb,yang2020neurocard,reiner2023sample,li2023alece}. These methods focus on using machine learning to estimate cardinalities, while still relying on traditional hand-designed cost estimators and dynamic programming (DP)-based plan enumeration to generate the execution plan.
    \item \textbf{Pure learning-based cost estimator} \cite{marcus2019plan,sun2019end,marcus2021bao}. In this kind of approach, machine learning models are used to directly estimate the execution cost for each plan, with a traditional plan enumerator used to determine the optimal plan.
    \item \textbf{Direct learning-based order generator} \cite{marcus2018deep,yu2020reinforcement,yang2022balsa,chen2023loger}. These methods bypass the cardinality and cost estimation stages entirely, using machine learning, particularly reinforcement learning, to directly generate the execution plan recursively.
\end{itemize}

However, these learning-based methods designed for relational databases cannot be directly applied to subgraph queries, as they typically assume a fixed table schema, while graph data is inherently more heterogeneous and schema-flexible.
\nop{
However, due to differences in data structures, these methods for relational databases cannot be directly applied to subgraph queries (graph queries usually have more complex joins).
}
To the best of our knowledge, apart from RLQVO \cite{wang2022reinforcement}, there is no work on learning-based optimization for subgraph queries. RLQVO formulates the matching order generation as a Markov Decision Process (MDP), 
iteratively assigning scores to vertices and selecting them accordingly.
However, RLQVO determines the scores solely from the representation of the query vertex, without considering the overall structure of subqueries. As a result, it occasionally produces suboptimal matching orders, leading to poor performance (\secref{sec:experiments}).

Another potential approach to matching order optimization relies on cardinality estimation combined with a hand-crafted cost model. Several learning-based methods for subgraph counting (NSIC \cite{liu2020neural}, LSS \cite{zhao2021learned}, NeurSC \cite{wang2022neural}, and GNCE \cite{schwabe2024cardinality}) have been proposed, which can be applied to estimate cardinality. However, few studies have integrated these techniques into query optimizers as core components. This is primarily due to several challenges that hinder their practical adoption in graph database systems:




\begin{itemize}[leftmargin=8pt]
    \item \textbf{High computational cost for online prediction.} For example, NeurSC \cite{wang2022neural} requires running a Graph Neural Network (GNN) on a filtered graph that can be as large as the original data graph. This makes it inefficient when estimating the cardinalities of numerous subqueries. Moreover, these methods are primarily designed to estimate the cardinality of an entire query graph rather than its subqueries. Consequently, employing these methods for subquery-level cardinality estimation introduces additional latency, further exacerbating the computational overhead.
    \item \textbf{Neglect of Subquery Relationships.} Given a query graph $Q$, the subgraph query optimizer needs to predict the cardinalities of multiple subqueries of $Q$ for the purpose of plan generation. However, existing subgraph counting approaches only focus on the cardinality estimation of a single subgraph query ($Q$ itself), neglecting the relationships among these various subgraph queries.
    \item \textbf{Expressiveness Limits of GNN.} Existing methods \cite{liu2020neural,zhao2021learned,wang2022neural,schwabe2024cardinality} primarily use message-passing neural networks (MPNNs) for graph feature extraction. However, MPNNs are limited in their expressiveness, as they do not surpass the capabilities of the 1-WL (1-dimensional Weisfeiler-Lehman) graph isomorphism test, which constrains the accuracy of these methods.
\end{itemize}

\subsection{Our Solutions}
\nop{In this paper, we propose a novel neural solution for optimizing subgraph queries: \textbf{\plainmethod}.}
In this paper, we propose \textbf{\plainmethod}, a novel neural optimizer for subgraph queries (short for \textbf{Neu}ral \textbf{S}ubgraph Queries \textbf{O}ptimizer).
Our method does not fit into any of the three existing categories of learning-based optimizers for relational databases as shown in \figref{fig:method_classification}; instead, it is more like a combination of them. Specifically, our approach simultaneously learns both the cardinality and the cost of queries and outputs a matching order under a new plan enumerator. The key insight behind jointly learning cardinality and cost is that both can be viewed as intrinsic properties of subqueries. This dual learning approach improves the robustness and accuracy of subgraph representations. 
We also introduce a new metric, \emph{minimum cost} (see \defref{def:minimumcost}), which refers to the execution cost of the best possible plan for each subquery among all potential plans.
Using the minimum cost, we propose a new top-down plan enumerator that significantly reduces the enumeration cost. 
Additionally, users can leverage the learned cardinality estimates to produce the matching order in traditional optimization settings. 

In summary, we make the following contributions in this paper:
\begin{itemize}[leftmargin=8pt]
    \item \textbf{Multi-task training framework.} 
    Our approach employs a \textit{multi-task} training framework. Since cardinality and execution costs are intrinsic properties of subqueries, we enhance the model's robustness and predictive accuracy by simultaneously training it to predict these two aspects.
    \item \textbf{Novel query graph encoder.} 
    We propose an enhanced GNN-based architecture that equips each vertex with richer neighborhood information compared to traditional MPNNs, thereby increasing the model's expressive power and predictive accuracy. Additionally, we improve the interaction between data graphs and query graphs. By initializing query vertex features with statistical information from filtered data graphs, our query graph encoder achieves both effectiveness and efficiency. 
    \item \textbf{Top-down subgraph query plan enumerator.} 
    We design a top-down greedy plan enumeration algorithm that generates high-quality execution plans. This approach minimizes excessive substate exploration typical in dynamic programming and outperforms heuristic-based methods in terms of performance. 
    \item \textbf{Extensive experiments on real and synthetic datasets.} 
    Extensive experiments on multiple real-world and synthetic datasets confirm that our method outperforms existing approaches in both efficiency and performance.
\end{itemize}

\section{Preliminaries}
\label{sec:preliminaries}

We focus on undirected vertex-labeled graphs in this paper. A graph $G$ is represented as a quadruple $G(V, E, L, \Sigma)$, where $V$ is the set of vertices and $E$ is the set of edges connecting vertices in $V$. $\Sigma$ is a set of some vertex labels, and $L$ is a function mapping every vertex in $V$ to a unique label from $\Sigma$. Sometimes, without causing ambiguity, we denote $G(V, E, L, \Sigma)$ as $G(V)$ or $G(V, E)$ for simplicity.

\tabref{tab:notations} lists the notations frequently used in this paper.

\begin{small}
\begin{table}[ht]
    \setlength{\abovecaptionskip}{0.1cm}
    \setlength{\belowcaptionskip}{-0.22cm}
  \caption{Frequently used notations.}
  \label{tab:notations}
  \begin{tabular}{ c | c } 
    \toprule 
    Notation & Description \\
    \midrule 
    $G(V, E, L, \Sigma)$ & undirected vertex-labeled graph \\
    $Q$, $q$ & query graph and subquery \\
    $o,Q_i$ & matching order and the partial query $Q(V_i^o)$ \\
    $N_o^-(u)$ & the backward neighbors of $u$ in order $o$ \\
    $C(u)$ & candidate set of the query vertex $u$ \\
    \bottomrule
\end{tabular}
\end{table}
\end{small}
\vspace{-0.29cm}

\subsection{Subgraph Query \& Subgraph Matching}


Subgraph query finds subgraph matches within a larger graph by identifying patterns that meet specific constraints. 
One of the most widely used matching semantic for subgraph matching is subgraph isomorphism \cite{prvzulj2006efficient,snijders2006new,chandramouli2010high,sahu2020ubiquity}, which is defined formally below:


\begin{definition}[Subgraph Matching (Isomorphism)] 
    Given a query graph $Q(V_Q, E_Q, L_Q)$ and a data graph $G(V_G, E_G, L_G)$, a subgraph matching is an injective function $f$ from $V_Q$ to $V_G$ such that
\begin{itemize}
    \item (Label constraint) $L_Q(v) = L_G(f(v))$, $\forall v \in V_Q$;
    \item (Edge constraint) $e(f(u), f(v))\in E_G$, $\forall e(u, v) \in E_Q$.
\end{itemize}
\end{definition}

We also refer to the mapped graph $m=f(Q)$ as a subgraph match of the query $Q$ in the graph $G$. The subgraph matching problem is to find all such subgraph matchings given a query and a data graph.

While our methods are generally applicable to various subgraph matching semantics, we adopt subgraph isomorphism here because of its practical importance and inherent complexity  \cite{prvzulj2006efficient,snijders2006new,chandramouli2010high,sahu2020ubiquity}. Extensions to other semantics (like homomorphism for SPARQL \cite{sparql} or cyphermorphism for Cypher \cite{francis2018cypher}), are discussed in \secref{sec:extension_matching}.

For simplicity of presentation, we assume that the query graphs are always connected. This does not affect the generality of our methods: for every unconnected query graph, the subgraph matching results are the Cartesian product of the matching results of its connected components with an additional check for injectivity.

\subsection{Execution of Subgraph Queries}
\label{sec:pre_execution_queries}

Executing a subgraph query in a graph database essentially involves finding all matchings of the query graph within the stored data graph \cite{zou2011gstore}. This process can be divided into three stages: filtering, planning, and enumeration \cite{sun2020memory} as shown in \algref{alg:subgraph_matching_execution}.

\vspace{-0.1cm}
\begin{small}
\begin{algorithm}
\DontPrintSemicolon
\caption{Subgraph query execution process \cite{sun2020memory}}
\label{alg:subgraph_matching_execution}
\KwIn{The data graph $G$, the query graph $Q$}
\KwOut{The subgraph matching results $R=R(Q, G)$}
\SetKwFunction{Enumerate}{Enumerate}


$C\leftarrow$ Generate candidate sets\;
$o=(o_1, o_2,\cdots, o_n)\leftarrow$  Generate a matching order\;
\Enumerate{$Q, G, o, \emptyset, 1$}\;

\BlankLine
\SetKwProg{myproc}{Procedure}{}{}
\myproc{\Enumerate{$Q, G, o, M ,i$}}{
	\If{$i = |o| + 1 $}{ 
            $R\leftarrow R\cup \{M\}$; \Return
	}
        $LC(u_{o_i}, M)\leftarrow$ Compute local candidates\;
    \ForEach{$v \in LC(u_{o_i},M)$}
    {
        \If{$\nexists\, u' \in V_Q, s.t.\ M[u']=v$} 
	    {
                $\Enumerate(Q, G, o, M\cup \{(u_{o_i}, v)\}, i+1)$\;
        
	    } 
    }
}
\end{algorithm}
\end{small}
\vspace{-0.1cm}


First, the engine analyzes the query graph and applies filters to identify a smaller data graph from which all matching results can be retrieved. Next, an execution plan, typically a matching order, is formulated by the database engine's optimizer. Finally, the executor enumerates all subgraph matchings to obtain the final results according to the matching order.


\nosection{Filtering.}
Filtering is an important technique in subgraph matching. It employs various restrictions to quickly exclude vertices and edges in the data graph that are not relevant to the query, thereby constructing a candidate set for each query vertex. Vertices that do not appear in a candidate set can be pruned during execution. The formal definition of a candidate set is provided below:

\begin{definition}[Candidate Set]
    Given a query graph $Q$ and a data graph $G$, the candidate set $C(u)$ of a query vertex $u\in V_Q$ is a subset of vertices in $G$, such that if there exists a subgraph matching $f$ of $Q$, then $f(u)\in C(u)$.
\end{definition}




\nosection{Plan Generation (Optimization).}
An execution plan is generated during the planning process. 
For subgraph queries, the main difference across different execution plans is the matching order (also known as the join order) defined in \defref{def:matching_order}. 
In this paper, we use matching order and execution plan interchangeably as they essentially mean the same.


\begin{definition}[Matching Order]
\label{def:matching_order}
    For a query graph $Q = Q(\{u_1, u_2, \cdots, u_n\})$, a matching order is a permutation of the query vertices: $o = (u_{o_1}, u_{o_2}, \cdots, u_{o_n})$.  Additionally, we denote the prefix $i$-subquery vertex set as $V_i^o = \{u_{o_1}, \cdots, u_{o_i}\}$.
\end{definition}

It is reasonable to require that the matching order preserves prefix connectivity. Otherwise, it will result in the Cartesian product. Formally, along with the matching order, each partial query $Q_i = Q(V_i^o) = Q(\{u_{o_1}, \cdots, u_{o_i}\})$ is a connected subquery, and the next vertex to match, $u_{o_{i+1}}$, is a neighbor of the set $V_i^o$. We denote the backward neighbors of $u$ as $N_o^-(u)$, which represents the set of neighbor vertices of $u$ with matching orders preceding $u$. 

\nosection{Enumeration.}
During the enumeration of the $i$-th vertex $u_{o_i}$, the backward neighbors of $u_{o_i}$ already have matches. Therefore, the local candidates for $u_{o_i}$ in the partial matching $M$ can be computed using the partial matching $M$ intersecting the neighbor lists from the matches of $N_o^-(u_{o_i})$ (line 7 of \algref{alg:subgraph_matching_execution}).

\subsection{Challenges in Subgraph Query Optimization}

In relational databases, the most common cost-based optimizer usually consists of three parts: a plan enumerator, a cost estimator, and a cardinality estimator. When users input a query, the plan enumerator generates several plans for the input query. Then, the cost estimator invokes the cardinality estimator for the cardinality estimation of subqueries, which helps estimate the cost of candidate plans. Then the plan with minimum estimated cost is transferred to the execution engine for execution. 

For subgraph queries, there are also some works like the traditional relational optimizers which enumerate the whole join space in a dynamic programming framework \cite{mhedhbi2019optimizing, yang2022gcbo, feng2023kuzu}. However, 
the larger number of vertices and edges of graph queries greatly enlarges the plan space, making dynamic programming optimizers impractical for real scenarios. For instance, a query graph with 24 vertices may have up to millions of connected subqueries. Dynamic programming methods would generate one plan for each subquery as an optimal subsolution, which can be extremely time-consuming.

Thus, most of the existing methods \cite{bonnici2013subgraph,sun2020rapidmatch,he2008graphs,bi2016efficient,bhattarai2019ceci,han2019efficient,kim2021versatile} rely on greedy heuristic strategies. However, they usually lack the detailed exploration of data characteristics and often lead to local optimum.

\section{Overview}
\label{sec:overview}



Our method, \method,  consists of three main components: a query graph encoder, a cardinality \& cost predictor, and a plan enumerator. 

\begin{figure*}
    \setlength{\abovecaptionskip}{0.1cm}
    \setlength{\belowcaptionskip}{-0.35cm}
    \centering    
    \includegraphics[width=\linewidth]{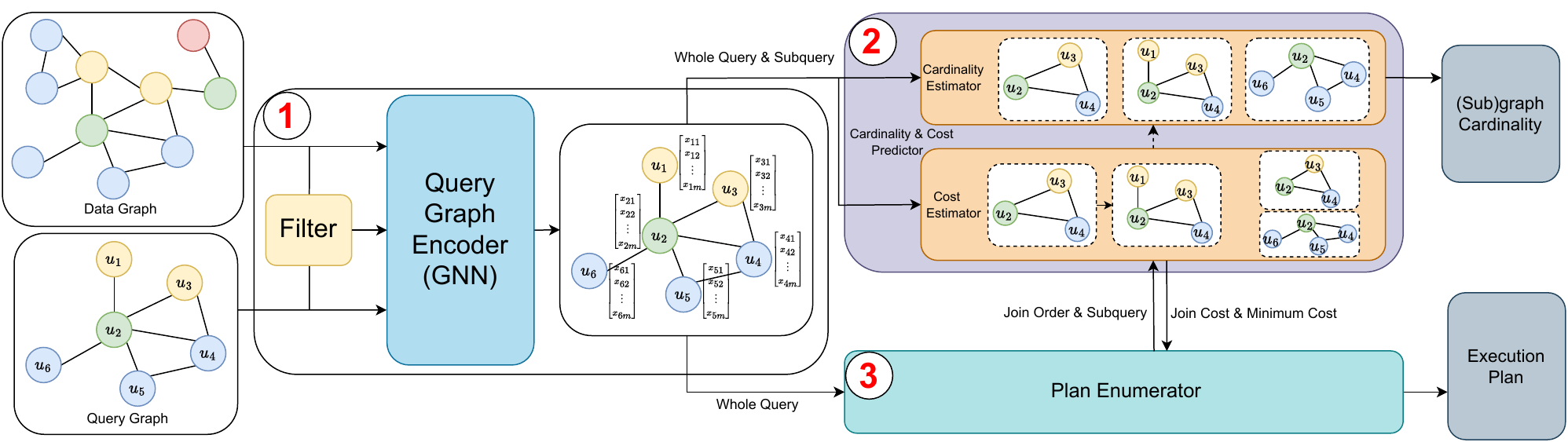}
    \caption{Overview of \method. 
    (1) GNN-based query graph encoder processes input after feature initialization from filtered data graph statistics to generate representation for (sub)query (\secref{sec:query_graph_encoder});
    (2) For each (sub)query, the cardinality \& cost predictor uses MLP to estimate the (sub)query cardinality and execution cost (\secref{sec:card_cost_estimator});
    (3) A top-down plan enumerator that utilizes the cost estimator to identify the execution plan with the lowest estimated cost (\secref{sec:plan_enumerator}).}    
    \label{fig:overview}
    \vspace{-0.1cm}
\end{figure*}

As illustrated in \figref{fig:overview}, the proposed method begins with a filter that generates query-related statistics upon receiving a new query graph. These statistics, along with the query graph, are processed by a GNN-based query graph encoder, which produces vertex-level representations. 
\nop{Subsequently, these vertex representations are aggregated to derive representations for each subquery and the entire query graph.}
Subsequently, these vertex representations can be aggregated to derive subquery representations on demand, as required by the plan enumeration phase. A representation for the entire query graph is also computed in this way.

Next, the cardinality \& cost predictor leverage these (sub)graph representations to estimate the cardinality and execution cost of the queries. Finally, a plan enumerator module utilizes these estimates to generate join plans by minimizing the cost estimates.


\method differs from existing learning-based optimizers, as illustrated in \figref{fig:method_classification}. After encoding the query graph, \method employs a multi-task learning framework to simultaneously predict both the cardinality and execution cost using the same subquery representation. Multi-task learning has gained significant attention recently and demonstrated effectiveness across various domains, such as NLP \cite{caruana1997multitask, zhang2021survey}. In the context of subgraph query optimization, this approach offers the following four key advantages:
\begin{itemize}[leftmargin=8pt]
    \item \textbf{Enhanced Representation Learning.} By sharing the same subquery representation, the model can inherently capture the relationships between cardinality and cost. This leads to improved performance through higher-quality subquery representations. 
    \item \textbf{Improved Robustness.} The multi-task framework allows the model to leverage supervisory signals from one task to mitigate noise or errors in the other, resulting in greater overall robustness. 
    \item \textbf{Increased Interpretability.} It is reasonable that the model produces a large cost estimate when the cardinality estimate is also large. Additionally, the intermediate cardinality and cost estimates provide insights that can help diagnose issues when the model produces suboptimal execution plans. 
    \item \textbf{Flexibility to Meet Diverse Requirements.} This multi-task framework is adaptable to different application needs. For instance, users with a well-designed cost estimator can utilize only the learned cardinality estimator, while those with less expertise and seeking an end-to-end solution can rely on the model's learned cost predictor. This flexibility also enables seamless integration with existing components.
\end{itemize}


\section{Query Graph Encoder}
\label{sec:query_graph_encoder}


In \method, we first utilize a graph neural network to encode the query graph, leveraging query-related information from the data graph (the query graph encoder in \figref{fig:overview}) to generate representations for (sub)queries. Details on the initial feature construction and the specific graph neural network used are provided in \secref{sec:feature_init} and \secref{sec:triangle_gnn}, respectively. Each query vertex is represented by a high-dimensional embedding after the encoding process.

Subsequently, as discussed in \secref{sec:subgraph_representation}, an aggregation pooling derives a representation for the (sub)query graph, serving as the input of the subsequent optimization process.

\subsection{Feature Initialization}
\label{sec:feature_init}

The design of initial features for an encoder is crucial. For a query graph, the most basic feature of a query vertex is its label. Given the relatively limited number of labels, one-hot encoding is widely adopted in existing subgraph tasks \cite{wang2022neural, ye2024efficient}. However, this encoding method presents two significant drawbacks: (1) it results in highly sparse representations that are challenging for subsequent learning processes; (2) it fails to incorporate insights from the topological structures of the data graph. 

An embedding-based approach effectively alleviates these issues. To obtain informative embeddings of each label in the data graph $G$, we pre-trained on a \textit{label-augmented} graph $G_A = G_A(V_G\cup V_L, E_G\cup E_L)$, which adds a vertex $l$ for each label in the label set $\Sigma_G$ ($V_L$) and an edge $e(u, L(u))$ for every vertex $u$ and its label $L(u)$ ($E_L$).  
Any existing embedding method \cite{perozzi2014deepwalk, grover2016node2vec, dong2017metapath2vec} can be applied to this label-augmented graph, and the resulting embedding $\bm{x}_l$ for a label vertex $l$ is utilized as the initial feature for the query vertex with label $l$.
In our experiments, we employ ProNE \cite{zhang2019prone}, which integrates sparse matrix factorization with embedding propagation, offering a fast and effective solution for downstream tasks.

It should be noted that relying solely on label features is inadequate, as it assigns identical features to vertices with the same label, regardless of whether they belong to the same query graph or different ones. To incorporate more query-specific information, we enhance the initial label features by integrating filtered statistics.
As described in \secref{sec:pre_execution_queries}, the filter stage produces a candidate set $C(u)$ for each query vertex $u$. Additionally, the candidate edge count $|C(u_1, u_2)|$ can be computed for each query edge $e(u_1, u_2)$ with a very small cost. These statistics are more accurate and query-relevant than those derived directly from the original data graph.

In summary, the initial features of the query graph encoder are constructed for a vertex $u$ and an edge $e(u_1, u_2)$ as follows:
\begin{equation}\label{eqn:vertex_feature}
    \bm{x}_u^{(0)} = \bm{x}_{L(u)} \oplus |C(u)|,
\end{equation}
\begin{equation}\label{eqn:edge_feature}
    \bm{x}_{e(u_1, u_2)}^{(0)} = \bm{x}_{u_1}^{(0)}\oplus\bm{x}_{u_2}^{(0)}\oplus |C(u_1, u_2)|, 
\end{equation}
where $\oplus$ denotes concatenation.

\subsection{TriAT: Triangle Attention Network}
\label{sec:triangle_gnn}

\subsubsection{Motivation for TriAT}
\label{sec:motivation_triat}

Graph neural networks (GNNs) are widely used to extract features and identify patterns in graph-structured data. Popular GNN such as GCN \cite{kipf2017semisupervised}, GraphSAGE \cite{ham2017graphsage}, GIN \cite{xu2018how} and GAT \cite{velivckovic2018graph} belongs to the message-passing neural network (MPNN) framework. These models iteratively update vertex embeddings as follows:
\begin{equation}
\label{eqn:MPNN}
    \bm{x}_u^{(k)} = \gamma^{(k)}\{\bm{x}_u^{(k-1)}, \phi^{(k)}\{\bm{x}_v^{(k-1)}, \bm{x}_{e(u, v)}^{(k-1)}\mid v\in N(u)\}\},
\end{equation}
where $\phi$ aggregates information from neighbors, and $\gamma$ combines it with the vertex’s prior representation $\bm{x}_u^{(k-1)}$ and edge representation $\bm{x}_{e(u, v)}^{(k-1)}$, producing the next-layer representation $\bm{x}_u^{(k)}$.


Despite their success, theoretical studies have shown that the expressive power of MPNNs is limited by the 1-WL test \cite{leman1968reduction,xu2018how}, preventing them from distinguishing certain graph topologies. For example, in \figref{fig:wl_test}, two graphs have distinct topologies but cannot be distinguished by MPNNs, as vertices with identical neighborhood label sets receive identical embeddings. For instance, vertices $A_1$ and $A_2$ in both graphs gather information from neighbors with the same label distributions, resulting in identical representations under the update rule in \eqnref{eqn:MPNN}.

\begin{figure*}[ht]
\setlength{\abovecaptionskip}{0.22cm}
\begin{minipage}[b]{7.7cm}
    \centering
    \begin{subfigure}[b]{0.44\linewidth}
        \centering
        \includegraphics[width=\linewidth]{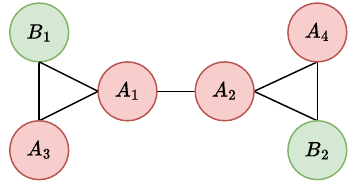}
        \caption{Graph 1}
        \label{fig:wl_1}
    \end{subfigure}
    \hspace{0.6cm}
    \begin{subfigure}[b]{0.4\linewidth}
        \centering
        \includegraphics[width=\linewidth]{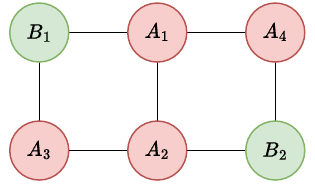}
        \caption{Graph 2}
        \label{fig:wl_2}
    \end{subfigure}
    \caption{Two graphs that cannot be distinguished by MPNNs or the 1-WL test. Vertices with identical notations receive the same representation under MPNNs.}
    \label{fig:wl_test}
\end{minipage}
\hfill
\begin{minipage}[b]{6cm}
    \centering
    \includegraphics[width=0.65\linewidth]{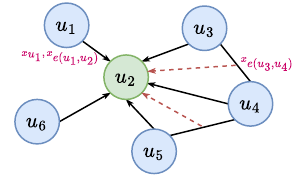}
    \caption{An example of the update procedure in TriAT (for $u_2$). Each vertex updates its embedding based on both its 1-hop neighbors and the edges between those neighbors.}
    \label{fig:triat_example}
\end{minipage}
\end{figure*}


Distinguishing between such graphs is crucial for subgraph query optimization, for different query graphs may lead to different optimal execution plans. To overcome this limitation, we propose \textbf{TriAT}, which explicitly incorporates triangle structures, enhancing the network's ability to capture higher-order patterns. In \figref{fig:wl_1}, triangles are formed by two $B$-label vertices linking to two $A$-label vertices, while no such triangle exists in \figref{fig:wl_2}. TriAT leverages these triangle structures to differentiate between graphs that standard MPNNs cannot.



It is worth noting that triangle patterns are ubiquitous in real-world graph structures and are commonly found in subgraph queries. As the simplest cyclic structure, triangles serve as a fundamental building block for many denser and more complex motifs, such as $k$-cliques. 
We observed that over $90\%$ of the cyclic queries tested in our experiments contain at least one triangle substructure. Previous research has also highlighted the prevalence of triangle structures in real graph queries \cite{bonifati2020analytical}. 
By identifying and leveraging these triangle patterns, we can effectively model complex relationships and higher-order dependencies, making them indispensable for robust graph representation.

\subsubsection{TriAT Architecture.}
\label{sec:trial_architecture}

TriAT updates the embedding of vertex $u$ by considering both its 1-hop neighbors $N(u)$ and the edges among these neighbors $E_{N(u)} = \{e(v_1, v_2)\mid v_1, v_2\in N(u)\}$, forming triangles with $u$. This is illustrated in \figref{fig:triat_example}, where vertex $u_2$ updates its embedding using both neighbor information (solid black arrows, such as $\{\bm{x}_{u_1}, \bm{x}_{e(u_1, u_2)}\}$ for neighbor $u_1$) and neighbor linkages (brown dotted arrows, such as $\{\bm{x}_{e(u_3,u_4)}\}$ for the edge $e(u_3, u_4)$).


Formally, TriAT follows the update procedure below:
\begin{equation}\label{eqn:triat_framework}
\begin{aligned}
    \bm{x}_u^{(k)} = \gamma^{(k)}\{\bm{x}_u^{(k-1)},\; &\phi^{(k)}\{\bm{x}_v^{(k-1)}, \bm{x}_{e(u, v)}^{(k-1)}\mid v\in N(u)\}, \\ &\tau^{(k)}\{\bm{x}_{e(v_1, v_2)}^{(k-1)}\mid e(v_1, v_2)\in E_{N(u)}\}\},
\end{aligned}
\end{equation}
\begin{equation}
\label{eqn:triat_edge_update}
    \bm{x}_{e(u, v)}^{(k)} = W_e^{(k)}(\bm{x}_u^{(k)}\oplus\bm{x}_v^{(k)}), (k\geq 1)
\end{equation}
here, $W_e^{(k)}$ is the learnable linear transformation matrix used for updating edge representations. $\phi^{(k)}$ and $\tau^{(k)}$ are attention-based aggregation functions for neighbor vertices and their interconnections:
\begin{equation}\label{eqn:triat_neighbor}
    \phi^{(k)}\{\bm{x}_v^{(k-1)}, \bm{x}_{e(u, v)}^{(k-1)}\mid v\in N(u)\} = \sum_{v\in N(u)}\alpha_{u,v}^{(k)}\Theta^{(k)} \bm{x}_v^{(k-1)},
\end{equation}
\begin{equation}\label{eqn:triat_neighbor_linkage}
    \tau^{(k)}\{\bm{x}_{e}^{(k-1)}\mid e\in E_{N(u)}\} = \sum_{e\in E_{N(u)}}\beta_{u,e}^{(k)}\Psi^{(k)} \bm{x}_{e}^{(k-1)}.
\end{equation}

$\Theta^{(k)}$ and $\Psi^{(k)}$ are learnable projection matrices. $\alpha_{u, v}^{(k)}$ and $\beta_{u, e}^{(k)}$ are attention coefficients computed as follows:
\begin{equation}\label{eqn:nei_attention}
    \alpha_{u,v}^{(k)}=\frac{\exp(\text{LR}(\bm{a}^{(k)}[W_1^{(k)}\bm{x}_u^{(k-1)}\oplus W_2^{(k)}\bm{x}_{e(u, v)}^{(k-1)}]))}{\sum_{w\in N(u)}\exp(\text{LR}(\bm{a}^{(k)}[W_1^{(k)}\bm{x}_u^{(k-1)}\oplus W_2^{(k)}\bm{x}_{e(u, w)}^{(k-1)}]))},
\end{equation}
\begin{equation}\label{eqn:nei_link_attention}
\begin{aligned}
    &\beta_{u,e}^{(k)}= &\frac{\exp(\text{LR}(\bm{b}^{(k)}[W_1^{(k)}\bm{x}_u^{(k-1)}\oplus W_2^{(k)}\bm{x}_{e(v_1, v_2)}^{(k-1)}]))}{\sum_{e'\in E_{N(u)}}\exp(\text{LR}(\bm{b}^{(k)}[W_1^{(k)}\bm{x}_u^{(k-1)}\oplus W_2^{(k)}\bm{x}_{e'}^{(k-1)}]))},
\end{aligned}
\end{equation}
where $\bm{a}^{(k)}$ and $\bm{b}^{(k)}$ are learnable attention weight vectors. $W_1^{(k)}$ and $W_2^{(k)}$ are learnable linear transformation matrices. 
LR denotes the LeakyReLU nonlinearity, and we set it with a negative slope of $0.2$ in our experiments. The attention mechanism in TriAT enables the model to assign varying weights to different neighbors and their connections based on their importance, enhancing its ability to capture complex patterns and dependencies.



For the update function $\gamma^{(k)}$, we sum the inputs and apply a ReLU nonlinearity. Additionally, we incorporate a multi-head mechanism \cite{vaswani2017attention} to further enhance TriAT's expressive power.

\subsubsection{Expressive Power \& Complexity}


The following theorems reveal the expressiveness and complexity of TriAT. Specifically, TriAT is more expressive than standard MPNNs and the 1-WL test, as it considers adjacency between neighboring vertices. Its time and space complexities are also characterized. 
\ifbool{fullversion}{The proofs of both results are provided in \appref{sec:proof_of_triat}.}{\textcolor{red}{The proof of \thmref{the:triat} can be found in the full version of our paper \cite{}.}}

\begin{theorem}[TriAT's Expressive Power]
\label{the:triat}
The following expressiveness inclusion relation holds: 1-WL test = MPNN $\prec$ TriAT.
\end{theorem}

\begin{theorem}[TriAT's Complexity]
\label{the:triat_complexity}
Let $\Delta$ denote the set of all triangles in a graph $Q(V, E)$. Then the time complexity of one layer of TriAT is $O(|E| + |\Delta|)$, and the space complexity is $O(|E|)$.
\end{theorem}

\ifbool{fullversion}{}{
\begin{proof}[Proof of \thmref{the:triat_complexity}]
Each layer of TriAT requires gathering information from neighbors ($\phi$) and their mutual connections ($\tau$) for every vertex.
The total computational complexity for $\phi$ and $\tau$ across all vertices is $O(|E|)$ and $O(|\Delta|)$, respectively.
This is because each edge is involved in $\phi$ through its two endpoints, and contributes to $\tau$ only if it is part of a triangle.
These two components together yield the overall time complexity of $O(|E| + |\Delta|)$. The space complexity is $O(|E|)$, as TriAT adopts a vertex-update framework that requires storing features for all vertices and edges.
\end{proof}
}



\nop{
In terms of complexity, each layer of TriAT requires gathering information from neighbors and their connections. The computational complexity of $\phi$ and $\tau$ for all vertices in a single layer is $O(|E|)$ and $O(|\Delta|)$, respectively, where $\Delta$ represents the set of all triangles in the graph. Consequently, the time complexity of one layer of TriAT is $O(|E| + |\Delta|)$, while the space complexity is $O(|E|)$, following the vertex-update framework of graph neural networks, which only requires updating the vertex features at each layer.
}

Although there are other GNNs that are more expressive than MPNNs, such as GNNs based on 2-FWL \cite{maron2019provably, zhang2023complete}, these methods incur a computational cost of $O(|V|^3)$ per layer and require $O(|V|^2)$ space. This high complexity significantly affects the runtime efficiency of 2-FWL. However, our experiments show that TriAT achieves comparable accuracy to 2-FWL without a substantial drop in performance, while maintaining a much lower computational cost (\secref{sec:ablation_model}). Therefore, TriAT provides a favorable trade-off between efficiency and effectiveness.

\subsection{Representation of (sub)Graphs}
\label{sec:subgraph_representation}

After applying the graph neural network to the query graph, we obtain an embedding matrix $X\in \mathbb{R}^{n\times m}$, where $n$ is the number of query vertices and $m$ is the vertex embedding dimension of the final layer of the TriAT used. Each row in the matrix $X$ corresponds to the embedding $\bm{x}_u^T$ of a query vertex $u$. To get the representation of the (sub)query graph $q$, we need to apply a pooling layer. 

It is important to note that some vertices may have a greater influence on the final result than others. For example, vertices with high degrees impose stricter topological constraints and prioritizing those vertices in matching order leads to lower cardinality. To capture the varying contributions of vertices, we apply a self-attention-based weighted pooling layer that assigns adaptive importance to each vertex in the final (sub)graph representation:
\begin{equation}\label{eqn:attention_pool}
    \bm{x}_q = \sum_{u\in V_q}\alpha_u \bm{x}_u,
\end{equation}
where, $\alpha_u$ is the attention coefficient for vertex $u$:
\begin{equation}
    \alpha_u = (K_1 \bm{x}_u)\cdot (K_2 \bm{x}_u),
\end{equation}
here, $K_1$ and $K_2$ are two trainable parameter matrices, and $\cdot$ denotes the inner product.

\section{Neural Graph Query Optimizer}
\label{sec:neural_optimizer}

After the query graph encoder, we can obtain representations for every subquery.
Next, we use learning-based estimators to predict cardinalities and execution costs and a top-down plan enumerator to generate matching order.

\subsection{Reformulation of Optimization}
\label{sec:reformulation}

\subsubsection{Intuition}
As discussed in \secref{sec:pre_execution_queries}, the execution plan of a graph query involves determining the matching order $o = (u_{o_1}, u_{o_2}, \cdots, u_{o_n})$ for the query graph $Q = Q(\{u_1, u_2, \cdots, u_n\})$. 
Along with the matching order $o$, each subquery $Q_i = Q(V_i^o) = Q(\{u_{o_1}, \cdots, u_{o_i}\})$ contains only one more query vertex $u_{o_i}$ than the previous subquery $Q_{i-1}$. 
Thus, the matching order can be seen as a state transition path from the initial state (empty query $Q_0$) to the final state (complete query $Q$). 
From this perspective, subgraph query optimization is the process of finding a near-optimal
state transition path within the linkage graph, consisting of all connected partial queries. 


\subsubsection{Formal Reformulation}

To start with, we define a lattice structure named \textbf{Cardinality-Cost Graph (CCG)} that captures the relationship between subqueries of one query graph $G$:

\begin{definition}[Cardinality-Cost Graph, CCG]
\label{def:ccg}
    Given a query graph $Q$, the CCG of $Q$ is a rooted (weakly) connected acyclic directed graph $CCG_Q(\mathcal{V}_Q, \mathcal{E}_Q, \mathcal{C}_Q)$, where:
    \begin{enumerate}[leftmargin=13pt]
        \item $\mathcal{V}_Q$ is the set of  connected partial subqueries of $Q$, including the empty subquery ($\emptyset$) and $Q$ itself. Each vertex $q$ in $\mathcal{V}_Q$ represents an intermediate state in the execution, referred to as a \textit{state}.
        \item $\mathcal{E}_Q$ is the edge set, where edges originate from smaller states and terminate at larger states that include one more query vertex. These edges represent \textit{state transitions}.
        \item $\mathcal{C}_Q$ is the cardinality and cost function. For each state $q \in \mathcal{V}_Q$, $\mathcal{C}_Q(q)$ denotes the actual cardinality of $q$, while for each state transition $e\langle q_{sub}, q(V_{q_{sub}} \cup \{u\})\rangle$, $\mathcal{C}_Q(e)$ indicates the actual execution cost of joining vertex $u$ to the partial intermediate results of $q_{sub}$. For convenience, we omit the subscript $Q$ when there is no ambiguity.
    \end{enumerate}
\end{definition}

Given a state $q_0$ from $CCG_Q$, we denote its out-neighbors as $\mathcal{N}^{out}(q_0) = \{q \in \mathcal{V}_Q \mid e\langle q_0, q\rangle \in \mathcal{E}_Q\}$. Similarly, we represent the in-neighbors of $q_0$ as $\mathcal{N}^{in}(q_0)$.

With the help of CCG, we can formulate the optimization as a shortest path problem on the CCG:

\begin{definition}[Optimization \& SP Problem on the CCG]
    Given a query graph $Q$, each join order $o = (o_1, \cdots, o_n)$ corresponds to a path $P$ from $\emptyset$ to $Q$ in $CCG_Q$:
    \begin{equation}
    \label{eqn:state_transition}
        P: (\emptyset =)\, q_0 \rightarrow q_1 \rightarrow q_2 \rightarrow \cdots \rightarrow q_n\, (= Q)
    \end{equation}
    with length defined as $l(P) = \sum_{i=1}^n \mathcal{C}_Q(e\langle q_{i-1}, q_i\rangle)$.

    The optimization's goal is to find the shortest path from $\emptyset$ to $Q$. 
\end{definition}

\nop{
\begin{theorem} 
\label{thm:shortest_path_optimal}
    The shortest path from $\emptyset$ to $Q$ is the join order leading to minimum execution cost.
\end{theorem}

\begin{proof}
    Each join order corresponds to a distinct path from $\emptyset$ to $Q$ in the CCG. For any given path $P=(q_0\rightarrow q_1\rightarrow\cdots\rightarrow q_n)$, its total execution cost can be decomposed into several one-step costs: $cost(P) = \sum_{i=1}^n \mathcal{C}_Q(e(q_{i-1}, q_i))=l(P)$. Consequently, minimizing execution cost is equivalent to finding the path $P^*$ with minimal length $l(P^*)$ in $CCG_Q$.
\end{proof}
}

\begin{figure*}[ht]
\setlength{\abovecaptionskip}{0.22cm}
\begin{minipage}[c]{6.7cm}
    \centering
    \begin{subfigure}[c]{0.3\linewidth}
        \centering
        \includegraphics[width=\linewidth]{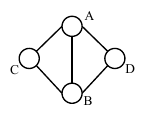}
        \caption{Query graph}
        \label{fig:query_example}
    \end{subfigure}
    \hspace{0.1cm}
    \begin{subfigure}[c]{0.45\linewidth}
        \centering
        \includegraphics[width=0.95\linewidth]{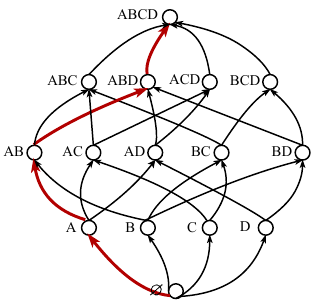}
        \captionsetup{skip=1.5pt}
        \caption{A CCG example}
        \label{fig:ccg_example}
    \end{subfigure}
    \vspace{-0.1cm}
    \caption{An example of a query graph and its corresponding CCG. For simplicity, we omit the cardinalities of states and the transition costs in the CCG.}
    \label{fig:ccg_example_with_query}
\end{minipage}
\hfill
\begin{minipage}[c]{6.7cm}
\begin{small}
\SetAlgoSkip{smallskip}
\begin{algorithm}[H]
\DontPrintSemicolon
\caption{Plan enumeration process}
\label{alg:plan_search}
\KwIn{The data graph $G$, the query graph $Q$}
\KwOut{The execution plan $o$}
$o\leftarrow \emptyset$, $q\leftarrow Q$\;
\While{$q\neq\emptyset$} {
    $q'\leftarrow \mathop{\arg\min}\limits_{q_0\in \mathcal{N}^{in}(q)} \hat{\mathcal{C}}(e\langle q_0, q\rangle)+\hat{\mathcal{MC}}(q_0)$\;
    $o.prepend(V_q-V_{q'})$\;
    $q\leftarrow q'$\;
}
\Return{o}\;
\end{algorithm}
\end{small}
\end{minipage}
\end{figure*}

\begin{example}
Consider the query graph in \figref{fig:query_example}. The topology of its corresponding CCG is depicted in \figref{fig:ccg_example}, where each state is annotated with the vertices of its associated subquery. And a possible vertex order $(A, B, D, C)$ is represented in the red path.


\end{example}

It is important to emphasize that the CCG is influenced by the underlying data graph, since the cardinality of each state and the cost of each transition are dependent on the specific data graph being matched.

If the cardinalities of the subqueries and the costs of state transitions are provided, a CCG for the query can be constructed. On this CCG, a shortest path algorithm is applied to determine the optimal execution plan. However, two major challenges arise:

\begin{enumerate}[leftmargin=10pt]
    \item The true cardinalities and costs are not precisely available during the optimization phase prior to execution. Therefore, it is necessary to employ a cardinality estimator and a cost estimator. 
    \item For large queries, the CCG may grow excessively large, rendering comprehensive exploration computationally prohibitive. For example, the CCG of a query graph with $24$ vertices can encompass millions of states.
\end{enumerate}

Our methods, as detailed in \secref{sec:card_cost_estimator} and \secref{sec:plan_enumerator}, are designed to address these two issues.

\subsection{Learning-based Cardinality \& Cost Predictor}
\label{sec:card_cost_estimator}

To address the first challenge (unavailable precise cardinality and cost during optimization), we leverage the powerful fitting and learning capabilities of neural networks. Specifically, we note that the cardinality is inherently tied to a single subquery, whereas the join cost is determined by the relationship between two adjacent subqueries, as represented by the edges on the CCG. 


\nosection{Cardinality.} For any subquery $q$, we use a multi-layer perceptron (MLP) $MLP_{card}$, for cardinality estimation. This can be formally expressed as:
\begin{equation}
    \hat{\mathcal{C}}(q) = MLP_{card}(\bm{x}_q).
\end{equation}

\nosection{Cost.} For a state transition $e\langle q_1, q_2\rangle$ on the CCG, where $q_2$ is a (sub)query that adds a linked vertex $u_0$ to $q_1$, we can utilize the above cardinality estimation and calculate the cost by any cost model such as gCBO \cite{yang2022gcbo} as below:
\begin{equation}
\label{eqa:cost_est_gcbo}
    \hat{\mathcal{C}}(e\langle q_1, q_2\rangle) = \hat{\mathcal{C}}(q_1) \times \min_{u \in N(u_0)\cap V_{q_1}}\frac{|C(u, u_0)|}{|C(u)|}.
\end{equation}

This cost estimate reflects the maximum possible number of intermediate results for $q_2$. Specifically, for each intermediate result $f$ of $q_1$, the term $\min_{u \in N(u_0) \cap V_{q_1}} \frac{|C(u, u_0)|}{|C(u)|}$ provides an upper bound on the number of matches of $q_2$ that would be generated from $f$.

This type of traditional method takes the cardinality estimate as input. However, it often overlooks certain aspects (such as storage cost), and the inherent inaccuracies in input cardinality estimation would propagate to the cost estimates. Therefore, we recommend an end-to-end approach, using neural networks to learn the relationship between the cost and the representation of the state transition in \method:
\begin{equation}
\label{eqn:join_cost_mlp}
    \hat{\mathcal{C}}(e\langle q_1, q_2\rangle) = MLP_{cost}(\bm{x}_{q_1}\oplus \bm{x}_{q_2}),
\end{equation}
where $\bm{x}_{q_1} \oplus \bm{x}_{q_2}$ represents the concatenation of the embeddings for the start state $q_1$ and the end state $q_2$.





\subsection{Top-down Cost-based Plan Enumerator}
\label{sec:plan_enumerator}

To tackle the second challenge (large CCG), we introduce the concept of the \textit{minimum cost of states} to guide the plan exploration on CCG:

\begin{definition}\label{def:minimumcost}
    The minimum cost of a state $q$ in the CCG is the minimum length of paths from $\emptyset$ to $q$:
    \begin{equation}
        \mathcal{MC}(q) = \min_{p\,:\, p_0=\emptyset,\, p_{|V_q|}=q}l(p).
    \end{equation}
\end{definition}

\begin{example}
    Consider the CCG in \figref{fig:ccg_example_with_query}. There are 6 different paths from $\emptyset$ to $q_{\{A, B, D\}}$. Suppose the shortest path among these is $P^* = \{\emptyset, q_{\{A\}}, q_{\{A, B\}}, q_{\{A, B, D\}}\}$. Then the minimum cost of $q_{\{A, B, D\}}$ is $l(P^*)$.
\end{example}

In traditional optimizers, the minimum cost of a state $q$ is determined through dynamic programming (DP) enumeration, which is time-consuming. Instead, we use a neural network (MLP) to directly predict the minimum cost:
\begin{equation}
    \hat{\mathcal{MC}}(q) = MLP_{mc}(\bm{x}_q).
\end{equation}

Based on the estimated minimum cost, we design a top-down greedy search for plan enumeration, as described in \algref{alg:plan_search}.


The critical step is in line 3. A state $q$ can be reached from states $\mathcal{N}^{in}(q)$. For a state $q_0$ among them, it can be accessed via a path from $\emptyset$ of length $\mathcal{MC}(q_0)$, which is $q_0$'s minimum cost. Therefore, the sum $\hat{\mathcal{C}}(e\langle q_0, q\rangle) + \hat{\mathcal{MC}}(q_0)$ gives the estimated length of the shortest path from $\emptyset$ to $q$ through $q_0$. The shortest path is selected and the additional vertex $V_q - V_{q'}$ is added to the matching order.


By utilizing the enumeration procedure in \algref{alg:plan_search}, we avoid constructing the entire CCG for the query $Q$. Instead, only promising states and transitions are explored, thereby reducing the cost of plan enumeration.

\begin{example}
    Take the CCG of \figref{fig:ccg_example} as an illustration. The enumerator of \algref{alg:plan_search} first explores all subqueries of size $3$. If it determines that subquery $q_{\{ABD\}}$ lies on the optimal path (as guided by the equation in line 3 of \algref{alg:plan_search}), it infers that vertex $C$ is the last to be joined. It then examines $q_{\{AB\}}$, $q_{\{AD\}}$, and $q_{\{BD\}}$, continuing the top-down process until reaching the initial state ($\emptyset$).
\end{example}

\subsubsection{Complexity}

Let the time complexities of running the models $MLP_{cost}$ and $MLP_{mc}$ be $T_{cost}$ and $T_{mc}$, respectively. During each selection step (line 3 of \algref{alg:plan_search}), the time cost is $O(|V_Q|(T_{cost} + T_{mc}))$, since the number of possible child states is bounded by $O(|V_Q|)$. As each selection step chooses a single vertex, the while loop repeats $O(|V_Q|)$ times. Therefore, the total complexity of \algref{alg:plan_search} is $O(|V_Q|^2(T_{cost} + T_{mc}))$. If we assume both $T_{cost}$ and $T_{mc}$ are $O(1)$, then the complexity simplifies to $O(|V_Q|^2)$.

In comparison, traditional dynamic programming (DP) methods typically have a time complexity of $O(|\mathcal{E}_Q|)$, which is often exponential in terms of $|V_Q|$, due to the need to explore all possible states. Consequently, our approach provides a substantial improvement in efficiency, particularly for large graphs.

\nop{
\subsubsection{Discussions}  
The plan enumeration approach employed by \method follows a top-down paradigm, distinguishing it from reinforcement learning-based optimization techniques such as RLQVO, which rely on bottom-up action selection. The advantages of \method's methodology can be summarized as follows:  

\begin{itemize}[leftmargin=8pt]  
    \item \textbf{History Information Reusability.}  
    The top-down enumeration strategy enables the reuse of historical optimization results. For instance, if a subquery matches a previously processed query in the log, the system might leverage cached vertex ordering information to eliminate redundant computations, though this optimization is not implemented in our current work. 

    \item \textbf{State Cost as an Intrinsic Property.}  
    In \method, the minimum cost associated with a state (subquery) is defined as an intrinsic property of the state itself (or more precisely, of the state and the initial empty state, which remains consistent across all queries). In contrast, reinforcement learning methods like RLQVO evaluate states based on their relation to the final state (i.e., the complete query $Q$), which varies per query. Consequently, for states of identical dimensionality, \method's cost computation is more efficient and less computationally demanding, which also reduces the learning difficulty. 
\end{itemize} 

}

\section{Training Details \& Extensions}
\label{sec:training_and_extensions}

In \secref{sec:data_collection}, we will first introduce
the training data collection, followed by the training process in \secref{sec:training_process}.
Finally, we show that our approach can also be extended to other matching semantics and directed edge-labeled graphs in \secref{sec:extension}.

\subsection{Training Data Collection}
\label{sec:data_collection}



To collect the training data for a data graph $G$, we first obtain a representative set of query graphs $\mathcal{T}$, which can be selected from user logs or simulated based on query templates.
For each query in $\mathcal{T}$, we construct its corresponding complete CCG. We then traverse each state within the CCG, recording (1) the cardinalities of the states, (2) the execution costs (running times) of the transitions, and (3) the minimum cost associated with each state. These data are stored as the training dataset.

However, for large-scale queries with many vertices or those producing substantial intermediate results, exhaustive traversal of all CCG states becomes computationally infeasible. To address this challenge, we employ a partial data collection strategy for such queries.
Specifically, for a large or complex query $Q$, we first generate an optimized matching order $o$ using the optimizer from \cite{yang2022gcbo}, while other optimization methods can also be employed. Then we explore not only the states $q_i$ (where $i = 1, 2, \ldots, |o|$) but also their adjacent states $N_{CCG}^{out}(q_i)$. In this process, we only collect (1) the state cardinalities and (2) the transition costs.

We categorize queries based on the extent of exploration: queries that are fully explored are labeled as  $q.state = Full$, while those partially explored are labeled as $q.state = Partial$.
Additionally, to ensure computational efficiency, we impose memory and time thresholds for each transition (3 minutes per transition and a space limit of 100 million embeddings per state in our experiments).

\subsection{Training Process}
\label{sec:training_process}

For fully explored query graphs, we utilize cardinalities, costs, and minimum costs during training. In contrast, for partially explored query graphs, we only use cardinalities and costs, as minimum costs are unavailable. In both scenarios, loss coefficients $\alpha_{i}$, where $i \in \{card, cost, mc\}$, are used to balance the multiple training objectives (we set $\alpha_{i}$ as 0.4, 0.3 and 0.3 for $i\in \{card, cost, mc\}$ in our experiments). For these objectives, we apply the L2 loss over the log transformation which imposes higher weights on larger error over the average \cite{dutt2019selectivity}:
\begin{equation}
\label{eqn:mse_error}
\mathcal{L}_q(\hat{y}, y)=|\log\hat{y}-\log y|^2=|\log(\max\{\hat{y}/y, y/\hat{y}\})|^2.
\end{equation}



In addition to these three separate losses from each neural predictor, we introduce a supplementary \textit{constraint loss} term to ensure consistency between the single-step cost predictor ($MLP_{cost}$) and the minimum cost predictor ($MLP_{mc}$). This constraint loss penalizes any unrealistic minimum cost estimate for a state $q_0$ that are lower than the minimum single-step cost estimates:
\begin{equation}
\label{eqn:single_step_punishment}
\mathcal{L}_c(q_0) = ||\max\{0, \min_{q'\in \mathcal{N}^{in}(q_0)}\hat{\mathcal{C}}(e\langle q',\, q_0\rangle)-\hat{\mathcal{MC}}(q_0)\}||^2.
\end{equation}

In our training process, the constraint loss is applied to the explored states $Q_e$. Note that for fully explored query graphs, the set of explored states $Q_e$ includes all subqueries of $Q$.
\ifbool{fullversion}{The entire training process is summarized in pseudocode in Appendix \ref{sec:alg_of_training_process}.}{\textcolor{red}{A pseudocode summarizing the entire training process is included in the full version of our paper \cite{}.}}

\subsection{Extensions}
\label{sec:extension}

\subsubsection{Extension to Other Matching Semantics}
\label{sec:extension_matching}



Thanks to the generalization ability of neural networks, \method can be extended to other subgraph matching semantics with minimal changes. Specifically, only two adjustments are needed: (1) adapting the query execution engine to reflect the desired matching semantics, and (2) collecting training data accordingly so that the supervision aligns with the new semantic. The neural model will then learn to approximate the new outputs automatically.

For example, when using \method for subgraph counting and optimization of homomorphism,
the execution engine do not need to do the isomorphic check (line 9 of \algref{alg:subgraph_matching_execution}), and the training data of cardinality and cost should be the number of homomorphisms and the homomorphism execution time, respectively.



\subsubsection{Extension to Directed Edge-labeled Graphs}
\label{sec:extension_to_directed_graphs}

Extending \method to directed and edge-labeled graphs is straightforward. 
The most basic modification involves filtering and enumerating matching results according to the directions and labels of the query edges.

Additionally, to account for the directionality of edges in the query graphs, the out-edges and in-edges should be handled separately in the query graph encoder. 
Specifically, we modify \eqnref{eqn:triat_framework} to the following iterative process:
\begin{equation}
\begin{aligned}
    \bm{x}_u^{(k)} = \gamma^{(k)}\{\bm{x}_u^{(k-1)},\, &\phi_1^{(k)}\{\bm{x}_v^{(k-1)}, \bm{x}_{e\langle u, v\rangle}^{(k-1)}\mid v\in N_{out}(u)\}, \\ &\phi_2^{(k)}\{\bm{x}_v^{(k-1)}, \bm{x}_{e\langle v, u\rangle}^{(k-1)}\mid v\in N_{in}(u)\}, \\ &\tau^{(k)}\{\bm{x}_{e\langle v_1, v_2\rangle}^{(k-1)}\mid e\langle v_1, v_2\rangle\in E_{N(u)}\}\}.
\end{aligned}
\end{equation}

We also incorporate edge label information into the initial edge features, as a modification of \eqnref{eqn:edge_feature}:
\begin{equation}\label{eqn:dir_labeled_edge_feature}
    \bm{x}_{e\langle u_1, u_2\rangle}^{(0)} = \bm{x}_{u_1}^{(0)}\oplus\bm{x}_{u_2}^{(0)}\oplus \bm{x}_{L(e\langle u_1,u_2\rangle)}\oplus |C\langle u_1, u_2\rangle|,
\end{equation}
where $\bm{x}_{L(e)}$ denotes the embedding of the edge label $L(e)$.
Existing methods can compute a representation for edge labels, such as RDF2Vec \cite{ristoski2016rdf2vec} and Ridle \cite{weller2021predicting}. In our experiments (\secref{sec:kuzu_exp}), we simply use one-hot encodings based on edge labels, which also yield good performance.


\section{Experiments}
\label{sec:experiments}

\subsection{Experiments Setup}

\subsubsection{Datasets}

We select six real-world datasets for the experiments, as they are widely used in prior work on subgraph matching \cite{bi2016efficient, sun2020memory, sun2020rapidmatch} and subgraph counting \cite{zhao2021learned, wang2022neural}. They cover various domains, including biology (Yeast, HPRD), social networks (DBLP, YouTube), the web (EU2005), and citation networks (Patents). These datasets vary in terms of scale and difficulty (e.g., topology, density), and their statistics are provided in \tabref{tab:datasets}.

\begin{small}
\begin{table}[ht]
    \setlength{\abovecaptionskip}{0.1cm}
    \setlength{\belowcaptionskip}{-0.27cm}
    \caption{Datasets statistics}
    \label{tab:datasets}
  \begin{tabular}{ c | c | c | c | c } 
    \toprule 
    Dataset & $|V|$ & $|E|$ & $|\Sigma|$ & average degree \\
    \midrule 
        Yeast & 3,112 & 12,519 & 71 & 8.0 \\
        HPRD & 9,460 & 34,998 & 307 & 7.4 \\
        DBLP & 317,080 & 1,049,866 & 15 & 6.6 \\
        EU2005 & 862,664 & 16,138,468 & 40 & 37.4 \\
        YouTube & 1,134,890 & 2,987,624 & 25 & 5.3 \\
        Patents & 3,774,768 & 16,518,947 & 20 & 8.8 \\
    \bottomrule
\end{tabular}
\end{table}
\end{small}

\subsubsection{Query graphs.}
We use queries from an influential subgraph matching survey \cite{sun2020memory}.
These queries are generated by randomly extracting connected subgraphs from the data graph, following the approach adopted in previous work \cite{sun2012efficient, bi2016efficient, han2019efficient}. This will ensure every query has at least one embedding in the data graph.
The queries vary in the number of vertices, ranging from 4 to 32, with 1800 queries for each dataset. Specifically, each query set with $i$ query vertices, denoted as $Q_i$ ($i = 4, 8, 16, 24, 32$), includes two types of queries: $Q_{iD}$ (dense queries) and $Q_{iS}$ (sparse queries).
Following \cite{sun2020memory}, a query is classified as dense if its average degree exceeds 3; otherwise, it is considered sparse.
Each category contains 200 queries per dataset, except for the size-4 queries, which are not further subdivided into dense or sparse.

During the data collection phase for training, we fully explore the queries in $Q_4$ and partially explore those in $Q_8$, $Q_{16}$, and $Q_{24}$ (\secref{sec:data_collection}). The query sets with 32 query vertices are only used for testing.
For the training process, we randomly sample 80\% of the queries from $Q_4$, $Q_8$, $Q_{16}$, and $Q_{24}$ as the training set, with the remaining 20\% and $Q_{32}$ used for testing.



\subsubsection{Compared Methods}

We compare our method with several traditional graph query ordering methods: QSI \cite{shang2008taming}, GQL \cite{he2008graphs}, RI \cite{bonnici2013subgraph}, RM \cite{sun2020rapidmatch}, and DPiso \cite{han2019efficient}. The first four methods are static ordering methods, while the last one, DPiso, is a dynamic ordering method that selects the next matched query vertex based on the partial match. All of these methods have been shown to be effective in numerous experiments and outperform other not selected methods \cite{sun2020memory, zhang2024comprehensive}. We include RLQVO \cite{wang2022reinforcement} in our comparison and adopt its experimental settings. Additionally, we did not find any other learning-based ordering methods apart from it.

For a fair comparison, we keep the filter and enumeration components of the pipeline execution identical across all methods; the only difference lies in the matching order. We do not include other dynamic methods, such as VEQ \cite{kim2021versatile}, as it is almost identical to DPiso in terms of order, except for degree-one vertices, and its ordering requires a special enumerator.

For \method, we implement it in Python using the AdamW optimizer \cite{loshchilovdecoupled}. The initial feature dimension of the query graph encoder is 128, and the number of TriAT layers is 2, with each layer having a dimension of 64. We train the models for 100 epochs with a learning rate of 0.002. A learning rate scheduler is used to decays the learning rate by 20\% every 20 steps. The query execution part of our experiments is adapted from the code in \cite{sun2020memory}, 
which is implemented in C++ and stores the graph using adjacency lists.
The codes of \method can be found at \href{https://github.com/fyulingi/NeuSO}{https://github.com/fyulingi/NeuSO}.

\subsubsection{Setup}  

Our experiments are conducted on a server with 256GB RAM, running Ubuntu 20.04.5 LTS. The server has an Intel Xeon Gold 6326 2.90GHz CPU and an NVIDIA A100-PCIE-40GB GPU. We set a 350-second time threshold for each query. Each experiment is repeated three times, and the mean values are reported.
\ifbool{fullversion}{
Due to space limitations, we sometimes only present experimental results on a few selected datasets. Additional results can be found in Appendix \ref{sec:appendix_more_exp_results} of the supplementary materials.}

\subsection{Results for Optimization}
\label{sec:exp_opt}

\subsubsection{Enumeration Time Comparison}
\label{sec:exp_enumeration_time_comparison}

In this experiment, we use enumeration time as an indicator of the quality of the generated matching orders. To isolate the effect of matching orderings, we ensure that all methods use the same filtering mechanism (GQL) and enumeration method (QSI), eliminating the influence of other factors. This allows us to focus solely on performance differences arising from the various matching order strategies.

\begin{figure}
    \setlength{\abovecaptionskip}{0.08cm}
    \setlength{\belowcaptionskip}{-0.3cm}
    \centering
    \includegraphics[width=0.7\linewidth]{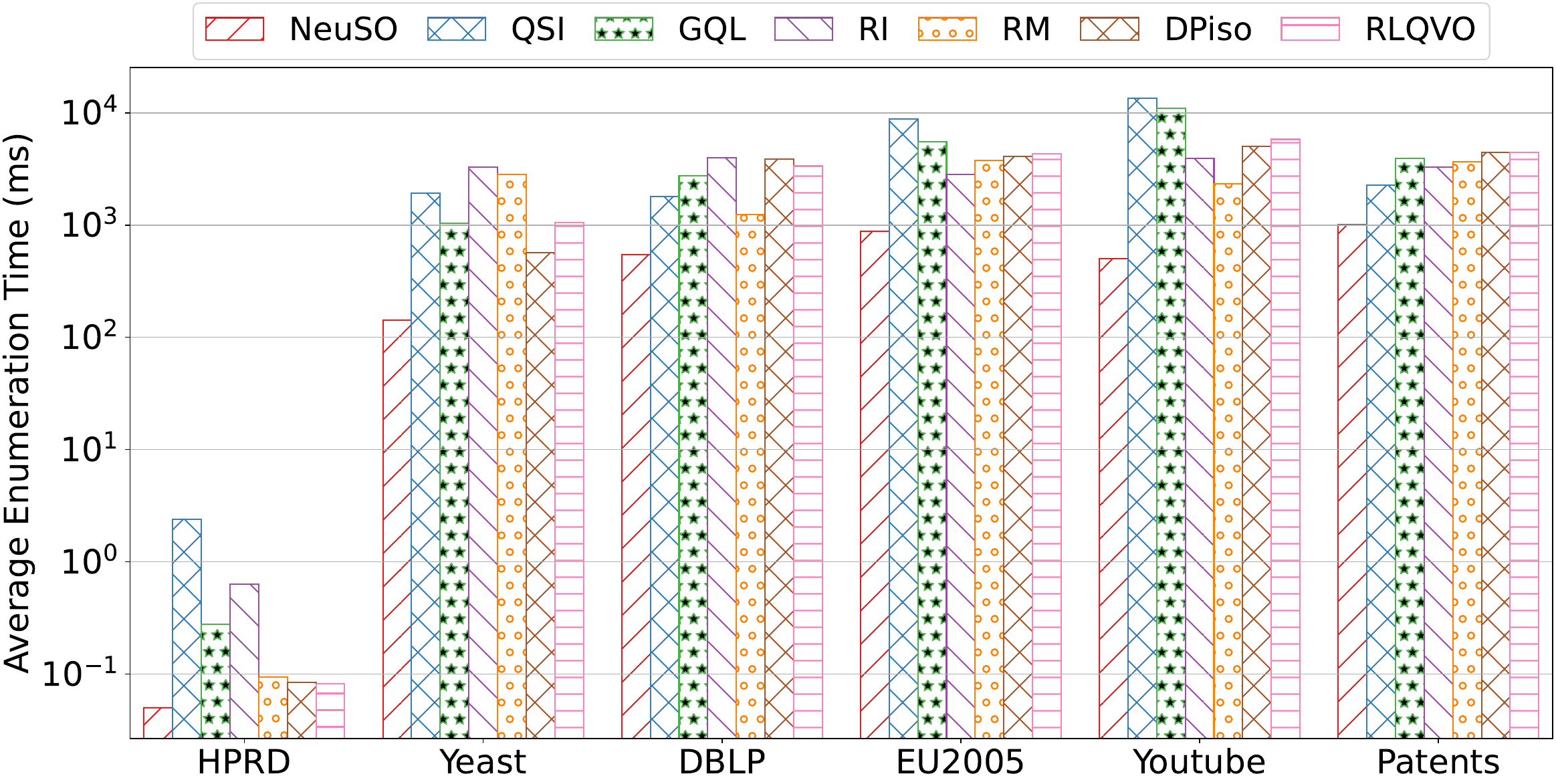}
    \caption{Average enumeration time comparison.}
    \vspace{-0.25cm}
    \label{fig:enumeration_all_datasets}
\end{figure}

\figref{fig:enumeration_all_datasets} presents the average enumeration time across the six datasets for different methods. Our results show that the matching orders generated by \method improve the enumeration efficiency compared to other methods by a factor of $1.63$ to $47.93$. This demonstrates that, after training, the model effectively learns optimization-relevant information from the data graph, enabling it to generalize its experience to unseen queries.

From \figref{fig:enumeration_all_datasets}, we observe that different matching orders result in significant variations in enumeration times. On the relatively simple HPRD dataset, all methods achieve low enumeration times, with the learning-based method RLQVO outperforming the other traditional methods. However, on more complex datasets such as DBLP and YouTube, RLQVO does not consistently outperform other methods and is 4.40 to 11.62 times slower than \method.
This performance gap can be attributed to several factors.
First, RLQVO is trained using reinforcement learning and may not have encountered certain transition patterns in the entire CCG graph during training. In contrast, \method is trained in a supervised manner using carefully collected training data, enabling it to learn more comprehensive mapping relationships between transition patterns and costs. 
Second, RLQVO employs GCN as its query graph encoder, which is less expressive and powerful compared to the TriAT encoder used by \method. 
Additionally, \method utilizes filtered statistics for feature initialization, which are more accurate than the naive statistics derived from the data graph used by RLQVO. 
Finally, RLQVO selects the next matching vertex by computing scores based on vertex representations, without considering higher-level subqueries, which limits its ability to optimize matching orders effectively.

\begin{figure}
    \setlength{\abovecaptionskip}{0.2cm}
    \centering
    \begin{subfigure}[c]{0.49\linewidth}
        \centering
        \includegraphics[width=0.9\linewidth]{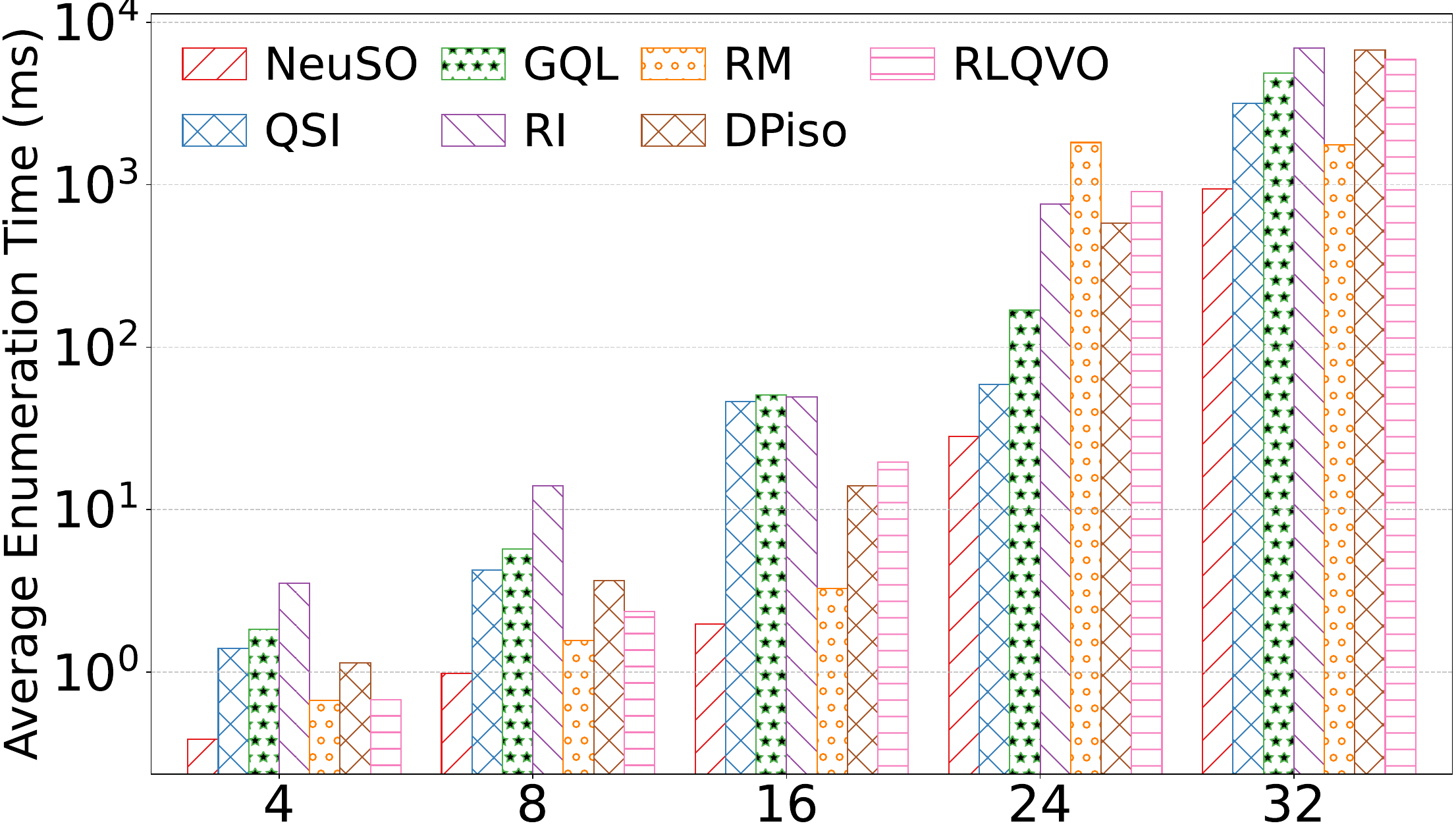}
        \captionsetup{skip=2pt}
        \caption{DBLP}
        \label{fig:dblp_enumeration}
    \end{subfigure}
    \begin{subfigure}[c]{0.49\linewidth}
        \centering
        \includegraphics[width=0.9\linewidth]{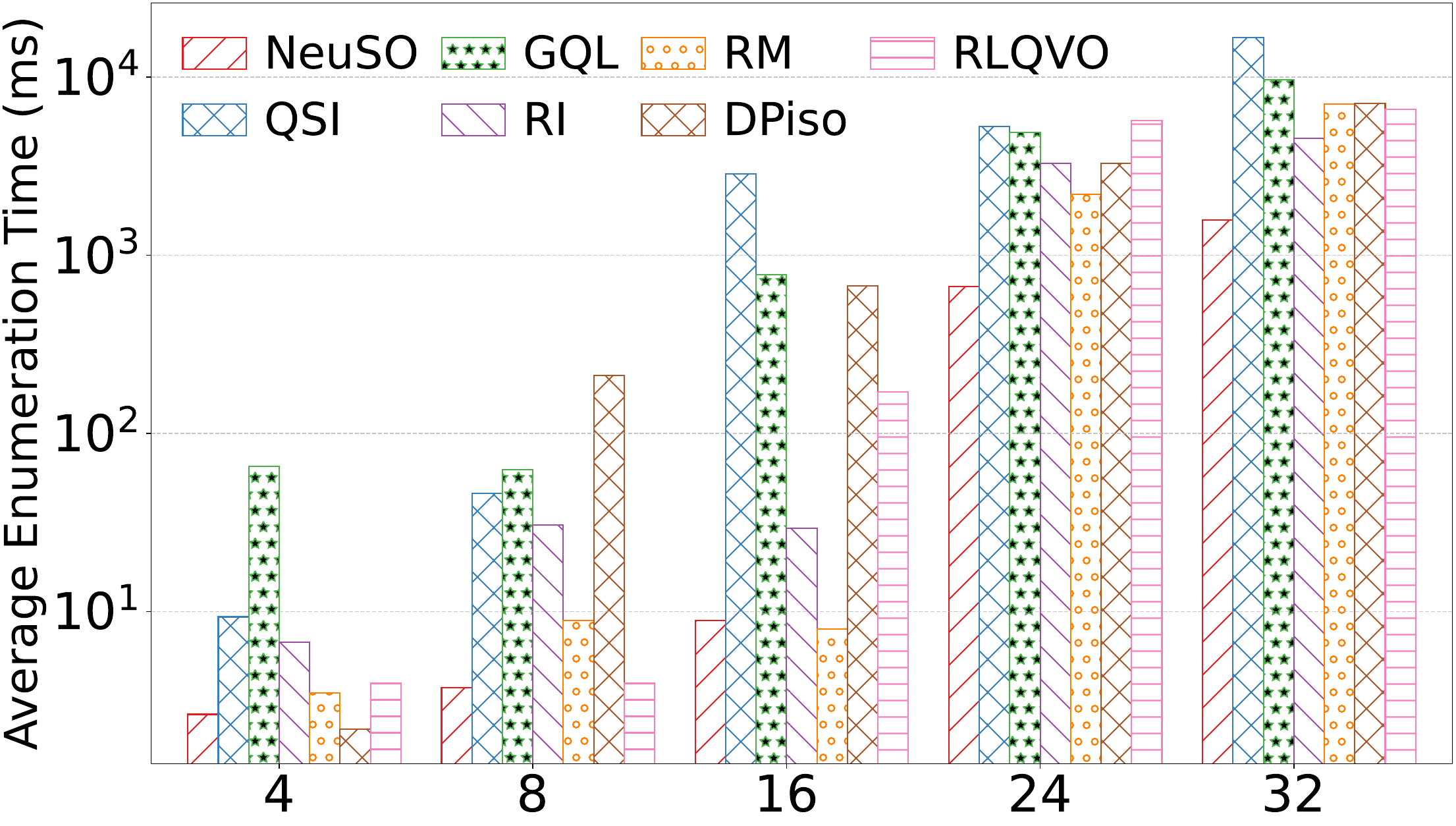}
        \captionsetup{skip=2pt}
        \caption{EU2005}
        \label{fig:eu2005_enumeration}
    \end{subfigure}
    \caption{Average enumeration time on DBLP and EU2005. The x-axis represents different query graph sizes $V_Q$.}
    \label{fig:dblp_eu2005_enumeration}
    \vspace{-0.2cm}
\end{figure}

\figref{fig:dblp_eu2005_enumeration} shows the enumeration time comparison on DBLP and EU2005 with different query graph sizes. Due to space limitations, we omit results for other datasets, which exhibit similar trends. 
The results indicate that \method maintains consistently strong performance across different query size.


Note that the average enumeration time can be significantly affected by extreme values, such as those very slow queries. To better understand \method's performance, we compared the enumeration time of each query on EU2005 between \method and the second-best method (RI). The results are shown in a scatter plot in \figref{fig:eu2005_scatter}. We find that \method accelerates queries especially that other methods take a long time to process. Most of these queries are large, with 24 or 32 vertices, which require a significant amount of time to retrieve results (some even take up to $10^5$ ms).
\figref{fig:eu2005_violin} present the log10 speedup of \method compared RI on EU2005. In general, the results show a long-tail distribution. The green point means that the average speedup is 2.58, indicating that when other methods choose a slow query plan, \method selects a better execution plan, leading to significant improvements.


\begin{figure}
    \setlength{\abovecaptionskip}{0.25cm}
    \setlength{\belowcaptionskip}{-0.15cm}
    \centering
    \begin{subfigure}[c]{0.36\linewidth}
        \centering
        \includegraphics[width=0.8\linewidth]{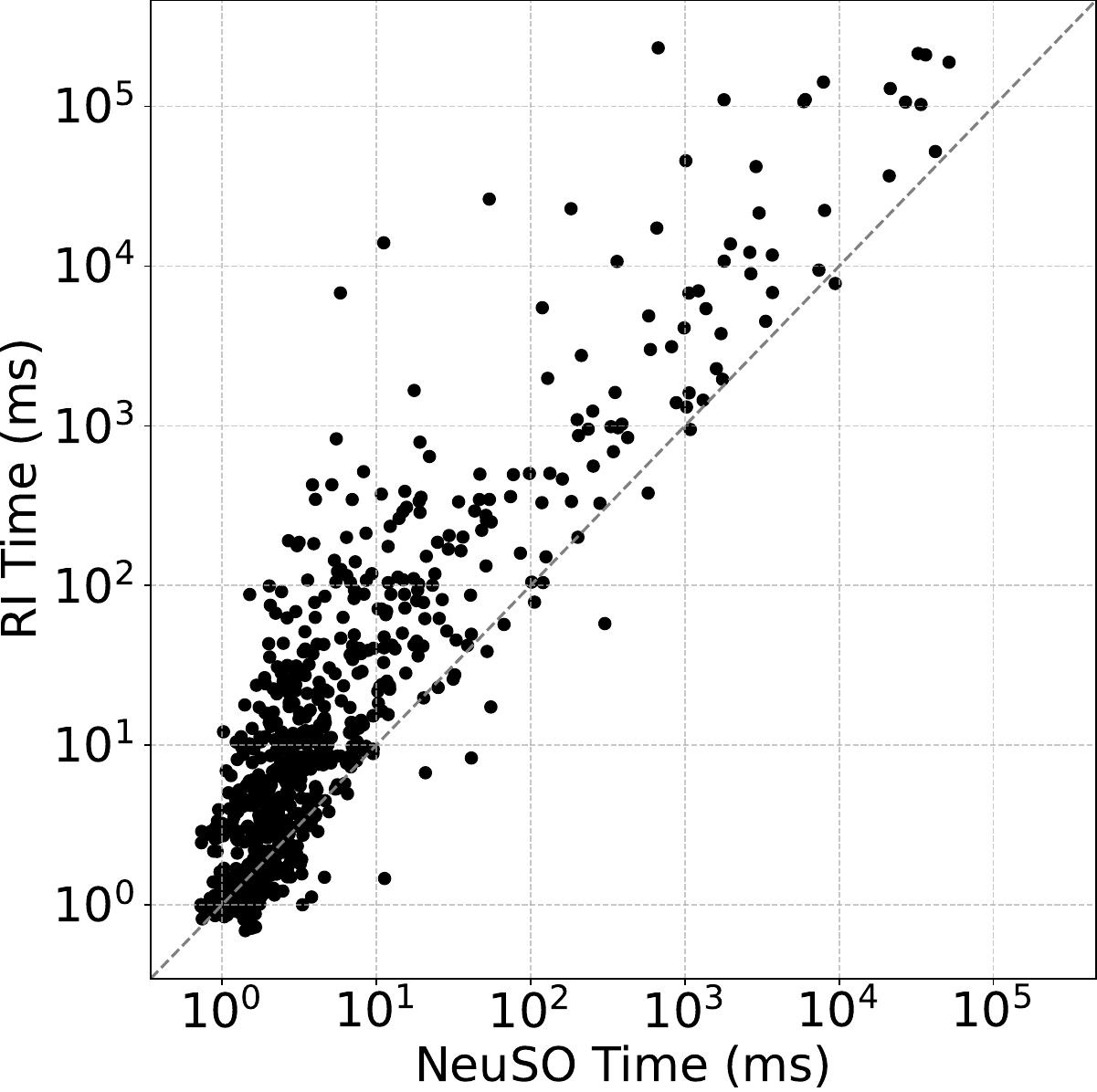}
        \caption{Enumeration time}
        \label{fig:eu2005_scatter}
    \end{subfigure}
    \hspace{0.2cm}
    \begin{subfigure}[c]{0.37\linewidth}
        \centering
        \raisebox{0.1\height}{
        \includegraphics[width=0.8\linewidth]{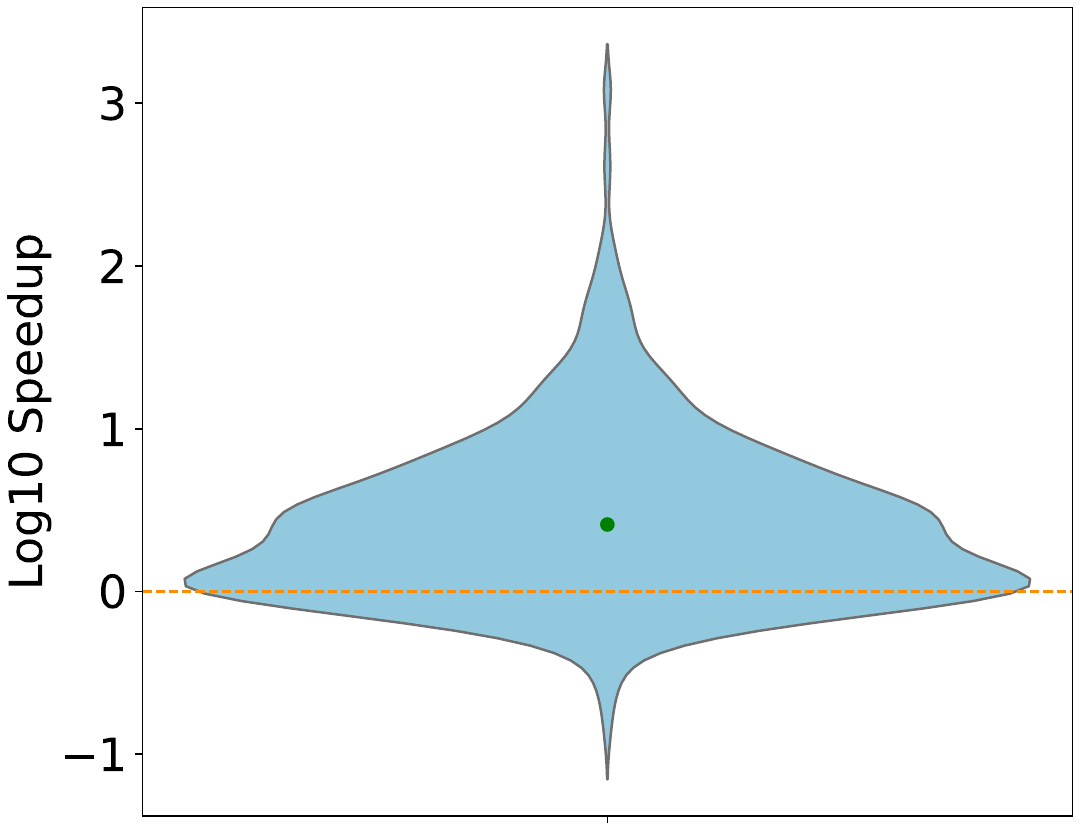}
        }
        \captionsetup{skip=2pt}
        \caption{Log10 of speedup}
        \label{fig:eu2005_violin}
    \end{subfigure}
    \caption{Enumeration time comparison on EU2005 (vs. RI).}
    \label{fig:eu2005_scatter_violin}
\end{figure}


\subsubsection{End-to-end Running Time Comparison}
\label{sec:end_to_end_time_comparison}

We designed \method to strike a balance between efficiency and effectiveness. \figref{fig:youtube_patents_end_to_end} compares the end-to-end running time for the YouTube and Patents datasets. While heuristic methods achieve faster optimization times, the gains from query optimization achieved by \method outweigh the additional time costs, resulting in superior overall performance.

\begin{figure}
    \setlength{\abovecaptionskip}{0.3cm}
    \setlength{\belowcaptionskip}{-0.3cm}
    \centering
    \begin{subfigure}[c]{0.49\linewidth}
        \centering
        \includegraphics[width=\linewidth]{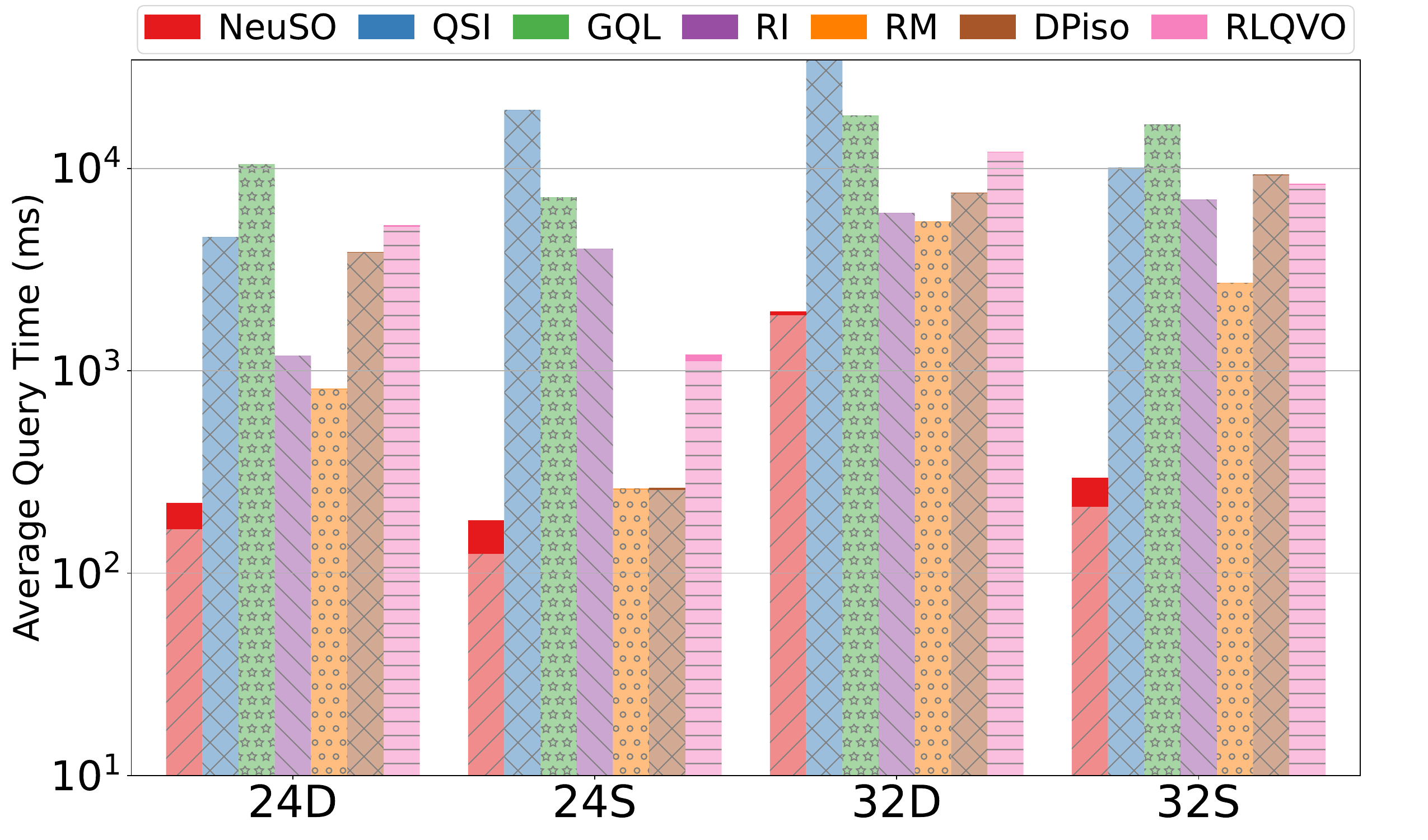}
        \captionsetup{skip=2pt}
        \caption{YouTube}
        \label{fig:youtube_end_to_end}
    \end{subfigure}
    \begin{subfigure}[c]{0.49\linewidth}
        \centering
        \includegraphics[width=\linewidth]{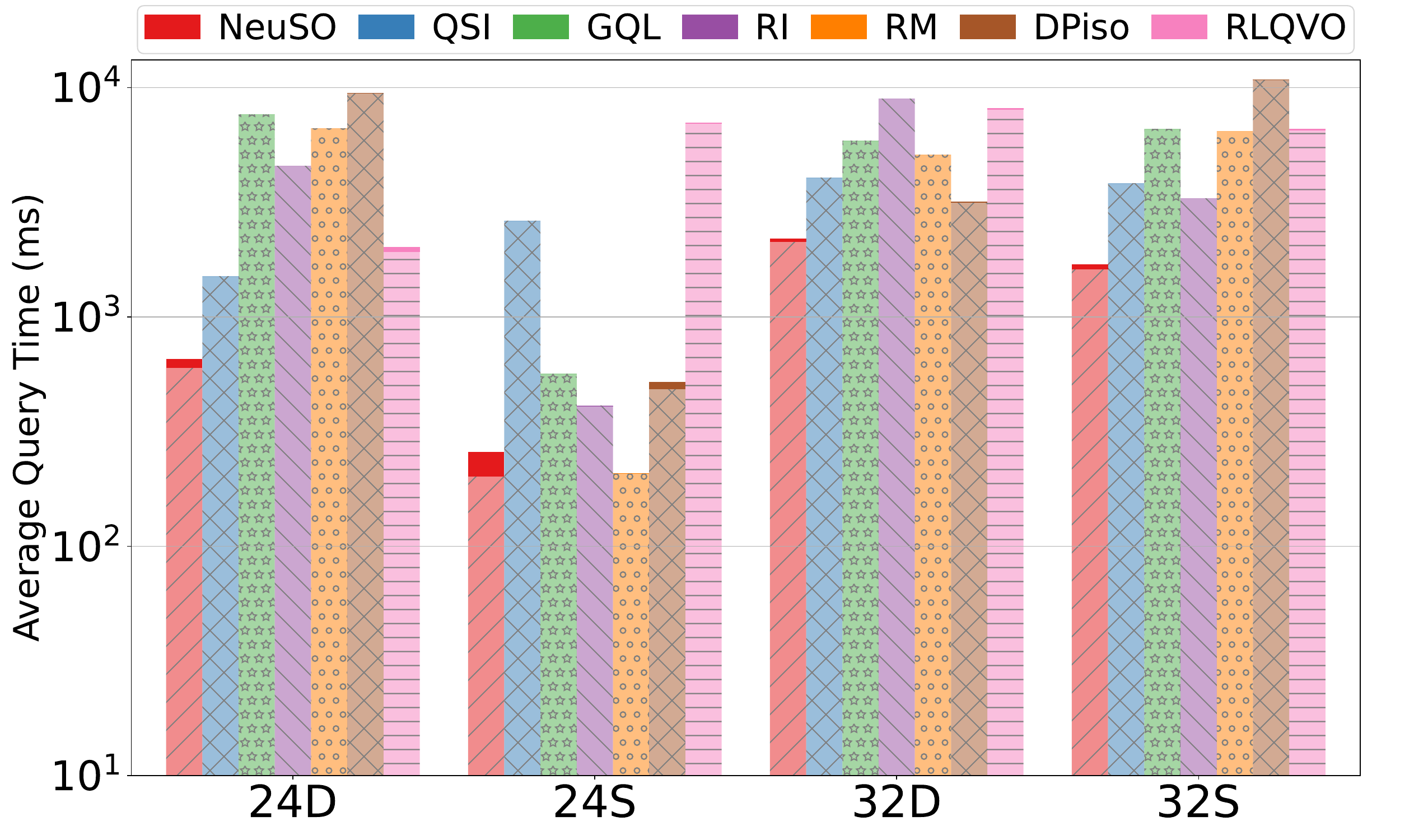}
        \captionsetup{skip=2pt}
        \caption{Patents}
        \label{fig:patents_end_to_end}
    \end{subfigure}
    \caption{Average end-to-end query time comparison. The bottom bars with hatching represent filter and enumeration time, while the solid bars above represent optimization times.}
    \label{fig:youtube_patents_end_to_end}
    \vspace{-0.2cm}
\end{figure}

For a query graph $Q$, \method requires $O(|V_Q|)$ decisions, each of which triggers one deep model invocation.
\tabref{tab:opt_time} compares the optimization time of \method and RLQVO. 
The efficiency advantage of \method primarily stems from the fact that \method runs only one GNN on the query graph, eliminating the need to re-run the query encoder after each decision, as required by RLQVO.

\begin{figure*}[ht]
\centering
\begin{minipage}[c]{0.45\textwidth} 
\begin{small}
\captionof{table}{Optimization time comparison (ms)}
    \vspace{-0.2cm}
    \label{tab:opt_time}
\resizebox{\linewidth}{!}{
  \begin{tabular}{  c | c | c | c | c | c  } 
    \toprule 
    Methods & $Q_4$ & $Q_8$ & $Q_{16}$ & $Q_{24}$ & $Q_{32}$ \\
    \midrule
        \plainmethod & 9.49 & 20.05 & 36.90 & 56.63 & 82.55 \\ \hline
        RLQVO & 12.94 & 29.35 & 61.18 & 93.30 & 125.48 \\
    \bottomrule
\end{tabular}
}
\end{small}
\end{minipage}
\hfill
\begin{minipage}[c]{0.52\textwidth} 
\begin{small}
\captionof{table}{Number of unsolved queries}
    \vspace{-0.2cm}
    \label{tab:unsolved}
\resizebox{\linewidth}{!}{
  \begin{tabular}{ | c | c | c | c | c | c | } 
    \hline
    Methods & Yeast & DBLP & EU2005 & YouTube & Patents \\
    \hline
        \plainmethod & \textbf{1} & 28 & \textbf{37} & \textbf{13} & \textbf{16} \\ \hline
        QSI & 6 & 42 & 96 & 41 & 24 \\ \hline
        GQL & 5 & 50 & 72 & 39 & 25 \\ \hline
        RI & 2 & 35 & 51 & 31 & 21 \\ \hline
        RM & \textbf{1} & \textbf{27} & 44 & 19 & 20 \\ \hline
        DPiso & 3 & 39 & 64 & 51 & 28 \\ \hline
        RLQVO & 9 & 58 & 48 & 45 & 34 \\
    \hline
\end{tabular}
}
\end{small}
\end{minipage}
\end{figure*}




\subsubsection{Unsolved Queries Number}

The number of unsolved queries is an important metric for evaluating the quality of the matching order. \tabref{tab:unsolved} presents the number of unsolved queries across five datasets, excluding HPRD, where all methods return results within the required time due to the simplicity of the dataset. A query is considered unsolved if it cannot be answered within 350 seconds. On most datasets, \method produces better matching orders, leading to a smaller number of unsolved queries compared to other methods.


\subsection{Ablation Study}

\subsubsection{Ablation on Multi-task Learning}

As described in \secref{sec:overview}, the multi-task learning framework of \method jointly trains the model to estimate both cardinality and cost. We hypothesize that this design improves the quality of subquery representations and enhances model robustness by enabling shared feature learning. To validate this hypothesis, we conduct the following ablation study.


\begin{figure}
    \setlength{\abovecaptionskip}{0.23cm}
    \setlength{\belowcaptionskip}{-0.25cm}
    \centering
    \begin{subfigure}[c]{0.49\linewidth}
        \centering
        \includegraphics[width=0.8\linewidth]{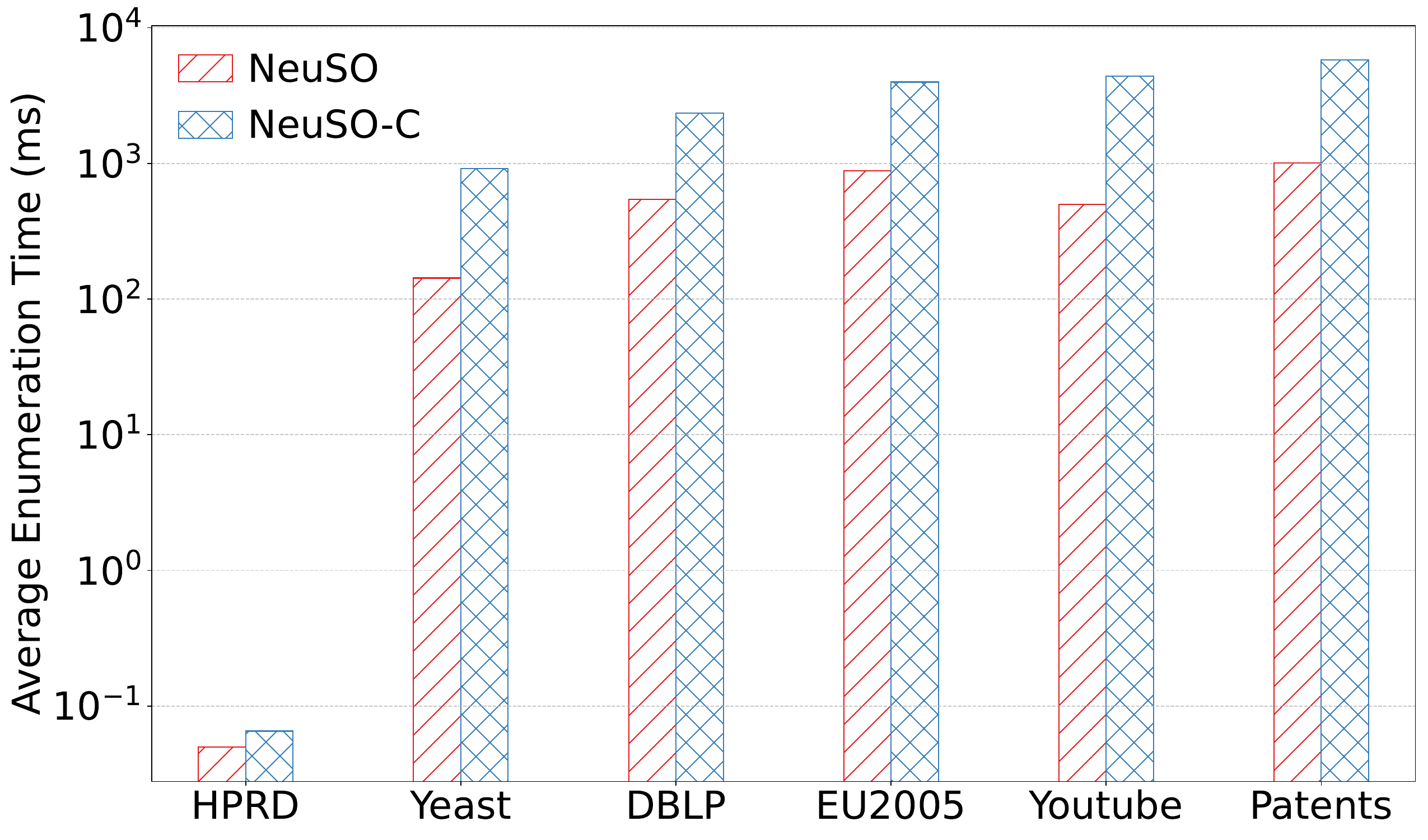}
        \captionsetup{skip=1.5pt}
        \caption{Ablation study on multi-task learning. }

        \label{fig:multi_task_ablation}
    \end{subfigure}
    \begin{subfigure}[c]{0.49\linewidth}
        \centering
        \includegraphics[width=0.8\linewidth]{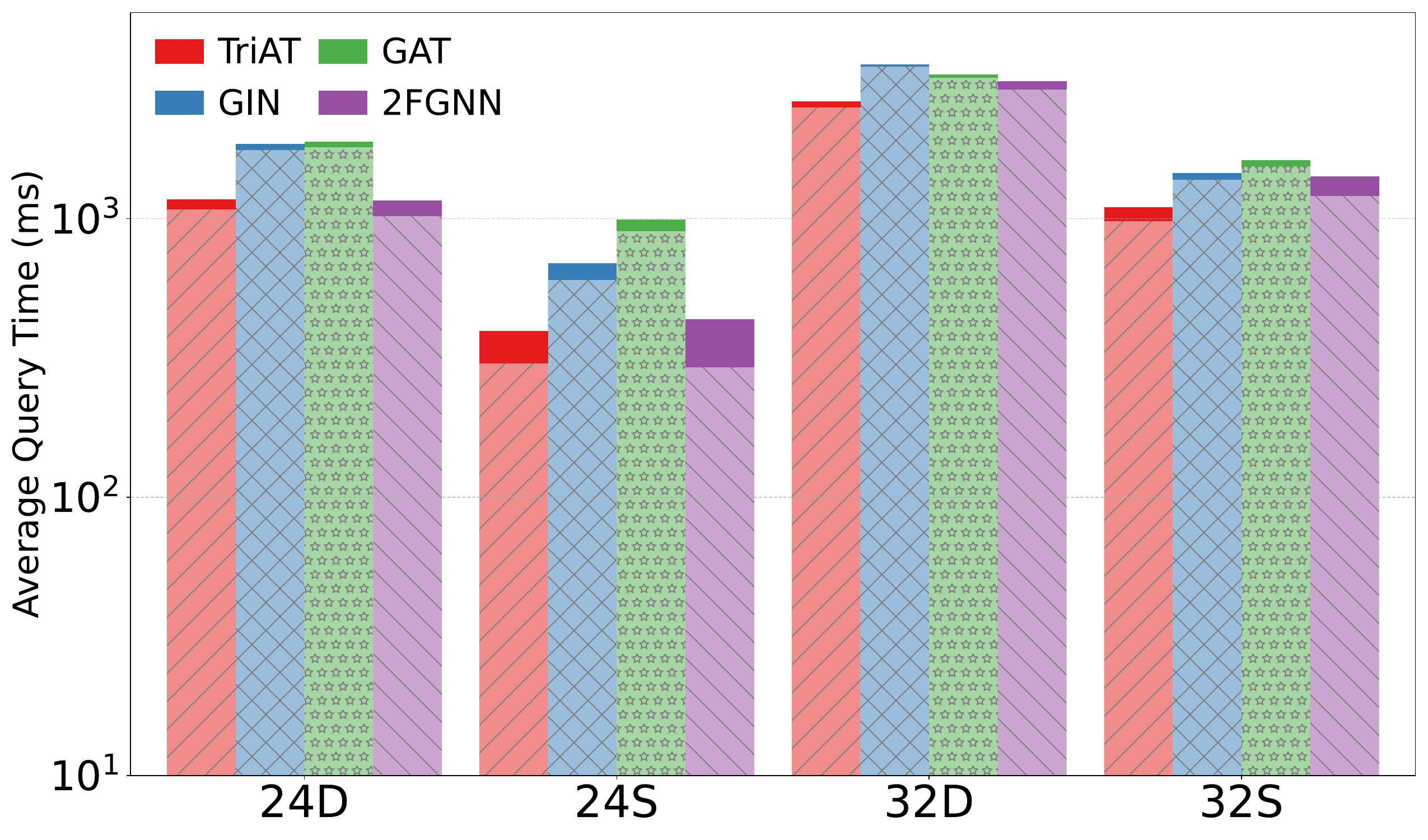}
        \captionsetup{skip=1.5pt}
        \caption{Ablation study on different GNN models (EU2005).}
        \label{fig:model_ablation}
    \end{subfigure}
    \caption{Ablation studies. (a) Enumeration time comparison with NeuSO-C (trained w/o cardinality estimation); (b) End-to-end performance of different GNN variants used in \method. The solid bars on top indicate optimization time.}
    \label{fig:ablation_exp}
    \vspace{-0.3cm}
\end{figure}

As shown in \figref{fig:multi_task_ablation}, compared to \texttt{NeuSO-C}, which is trained solely using subquery cost information (excluding cardinality), \method achieves a speedup in enumeration time ranging from $1.31\times$ to $8.78\times$. This result demonstrates the effectiveness and advantage of the multi-task learning framework.

\subsubsection{Ablation For GNN Model}
\label{sec:ablation_model}

To evaluate the effectiveness of our proposed TriAT module, we conduct an ablation study by replacing it with several representative GNN architectures, including 
GIN \cite{xu2018how}, GAT \cite{velivckovic2018graph}, and 2FGNN (based on 2FWL) \cite{maron2019provably}. 

The results on the EU2005 are shown in \figref{fig:model_ablation}. TriAT consistently outperforms GIN and GAT, benefiting from its stronger expressive power. It also achieves performance comparable to the theoretically more powerful model 2FGNN. However, 2FGNN incurs higher computational costs, resulting in weaker end-to-end performance for some large queries. We further observe that 2FGNN occasionally produces suboptimal plans, suggesting that excessive structural awareness does not necessarily translate to better optimization performance. All these results highlight the advantage of TriAT in balancing expressiveness and efficiency for query optimization tasks.

\vspace{-0.2cm}
\subsection{Results for Cardinality Estimation}
\label{sec:exp_ret_card}

\method is trained via a multi-task approach, with cardinality estimation as a byproduct (although it is not required for matching order generation in \method). We evaluate its effectiveness and efficiency compared to other methods.

We selected three state-of-the-art (SOTA) neural methods, LSS \cite{zhao2021learned}, NeurSC \cite{wang2022neural}, and GNCE \cite{schwabe2024cardinality}, as well as two traditional methods, Alley \cite{kim2021combining} and Fast \cite{shin2024cardinality}, for comparison.
These methods were chosen because they have significantly outperformed other cardinality estimation techniques. We used the official implementations provided by the authors and conducted experiments following the settings and parameters described in the original papers.

\subsubsection{Accuracy}

\nop{The q-error for the five compared methods is illustrated in \figref{fig:cardinality_accuracy}. Results are shown only for queries where cardinality can be accurately determined within a reasonable time frame. 
To distinguish between underestimation and overestimation, we plot the q-error as $\hat{c}/c$, where overestimation is depicted above $y=0$ and underestimation below $y=0$.}

The results of cardinality estimation accuracy 
are illustrated in \figref{fig:cardinality_accuracy}, where we only show the results for queries whose cardinalities can be accurately determined within a reasonable time frame.  
To distinguish between underestimation and overestimation, we plot the logarithmic estimation ratio between predicted and true cardinalities: $\log_{10}(\hat{c}/c)$, where overestimation is depicted above $y=0$ and underestimation below $y=0$.


\begin{figure}
    \setlength{\abovecaptionskip}{0.1cm}
    \setlength{\belowcaptionskip}{-0.3cm}
    \centering
    \includegraphics[width=0.88\linewidth]{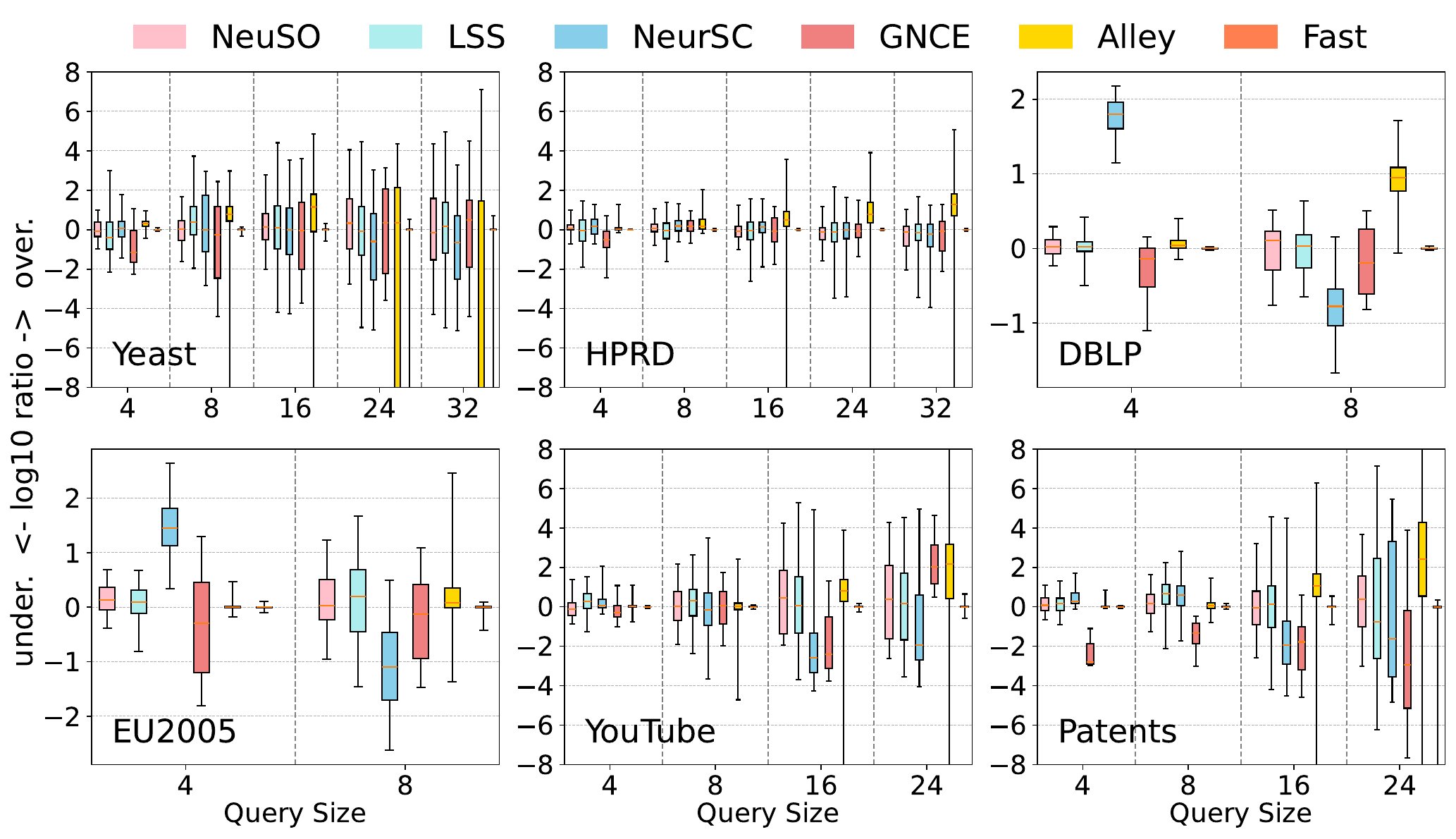}
    \caption{Cardinality estimation accuracy comparison. The boxplot shows $\log_{10}(\hat{c}/c)$, the base-10 log of the estimation ratio. The lower and upper whiskers denote the min and max values. The bounds of the box represent the 25th and 75th percentiles, with the median indicated by the line inside the box. A denser box closer to 0 indicates higher accuracy.}
    \label{fig:cardinality_accuracy}
    \vspace{-0.15cm}
\end{figure}

\nop{As depicted in \figref{fig:cardinality_accuracy}, \method consistently outperforms the two learning-based cardinality estimation methods, LSS and NeurSC. LSS provides relatively stable estimates, with balanced overestimations and underestimations. NeurSC performs better than LSS on the HPRD and Yeast datasets due to its filter effectively captures more query-relevant information from the data graph. , as filter effectively captures more query-relevant information from the data graph.
However, NeurSC performs poorly on other large datasets such as DBLP and EU2005, and tends to underestimate as the query size increases.}

As depicted in \figref{fig:cardinality_accuracy}, \method consistently outperforms the other three learning-based cardinality estimation methods: LSS, NeurSC, and GNCE. 
LSS provides stable estimates, with balanced overestimations and underestimations. 
NeurSC performs better than LSS on the HPRD and Yeast datasets, as filter effectively captures more query-relevant information from the data graph.
However, NeurSC performs poorly on larger datasets such as DBLP and EU2005, and tends to underestimate as the query size increases. 
Similarly, GNCE performs poorly on large datasets such as YouTube and Patents.



Alley performs well for small queries (e.g., 4 vertices), but deteriorates significantly as query size grows. This is primarily due to sampling failure: with over 16 vertices, Alley's method struggles to generate valid results, often leading to underestimation.

Fast, the SOTA sampling-based method, improves upon the sampling approach by sampling from the filtered data graph, increasing the likelihood of successful sampling. Our experiments show that Fast achieves the best performance in most cases. However, it also encounters sampling failures for large queries (e.g., size 32 on Yeast and size 24 on Patents). \method performs better than other methods except Fast in most scenarios, leveraging a more expressive GNN and utilizing query-specific information through its filter.


\subsubsection{Efficiency}
\label{sec:card_efficiency_exp}

Although Fast achieves higher accuracy than \method, its estimation cost is prohibitively high for query optimizers. \figref{fig:card_est_time_comp} compares the running time for a single estimate (entire query) on the EU2005 and DBLP datasets. 
For \method and NeurSC, the filtering time is included.

\begin{figure*}[ht]
\begin{minipage}[b]{6.7cm}
    \begin{subfigure}[c]{0.48\linewidth}
        \centering
        \includegraphics[width=\linewidth]{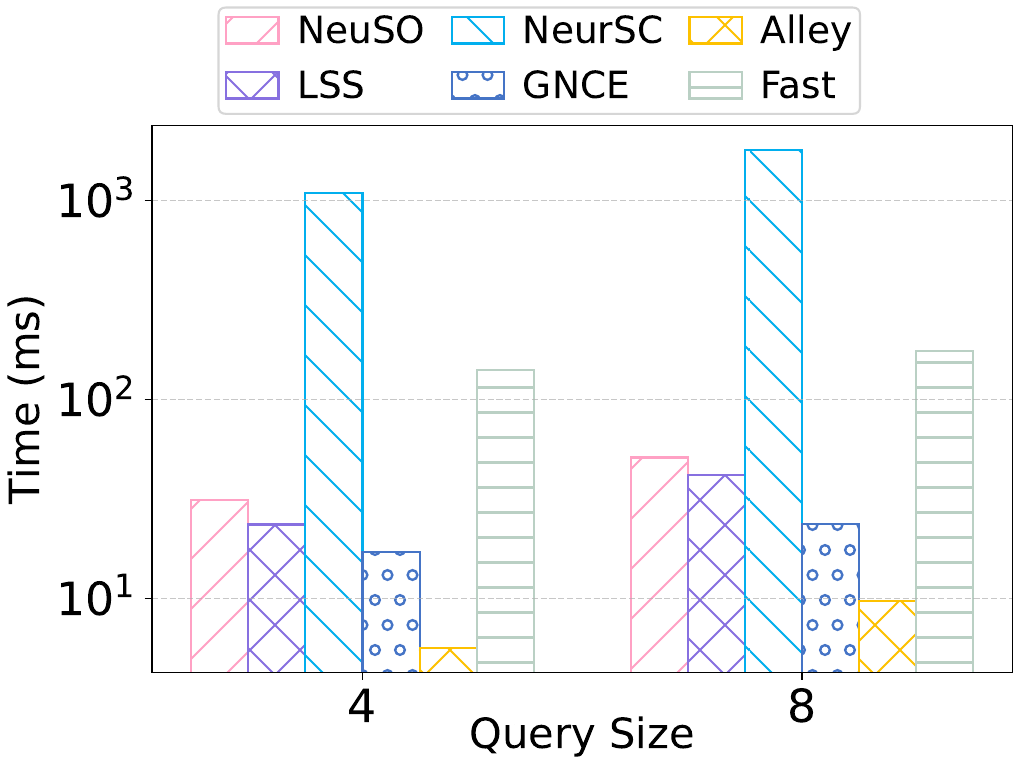}
        \captionsetup{skip=1.5pt}
        \caption{DBLP}
        \label{fig:card_time_dblp}
    \end{subfigure}
    \hfill
    \begin{subfigure}[c]{0.48\linewidth}
        \centering
        \includegraphics[width=\linewidth]{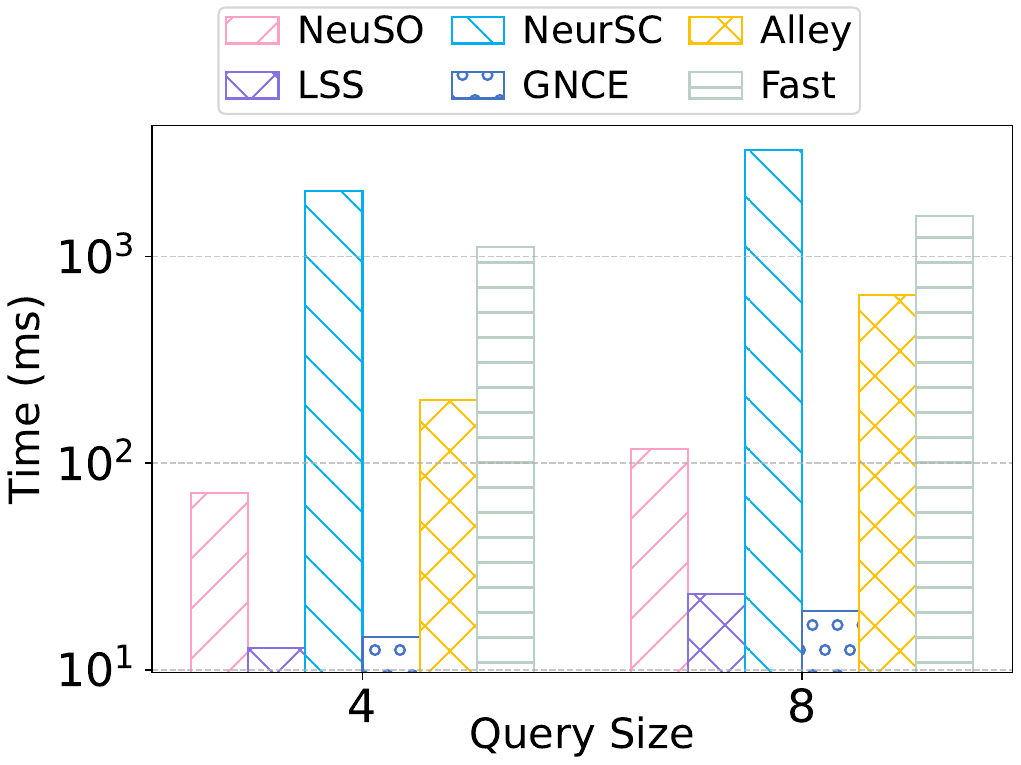}
        \captionsetup{skip=1.5pt}
        \caption{EU2005}
        \label{fig:card_time_eu2005}
    \end{subfigure}
    \caption{Average cardinality estimate time.}
    \label{fig:card_est_time_comp}
\end{minipage}
\hfill
\begin{minipage}[b]{7cm}
    \begin{subfigure}[c]{0.49\linewidth}
        \centering
        \includegraphics[width=\linewidth]{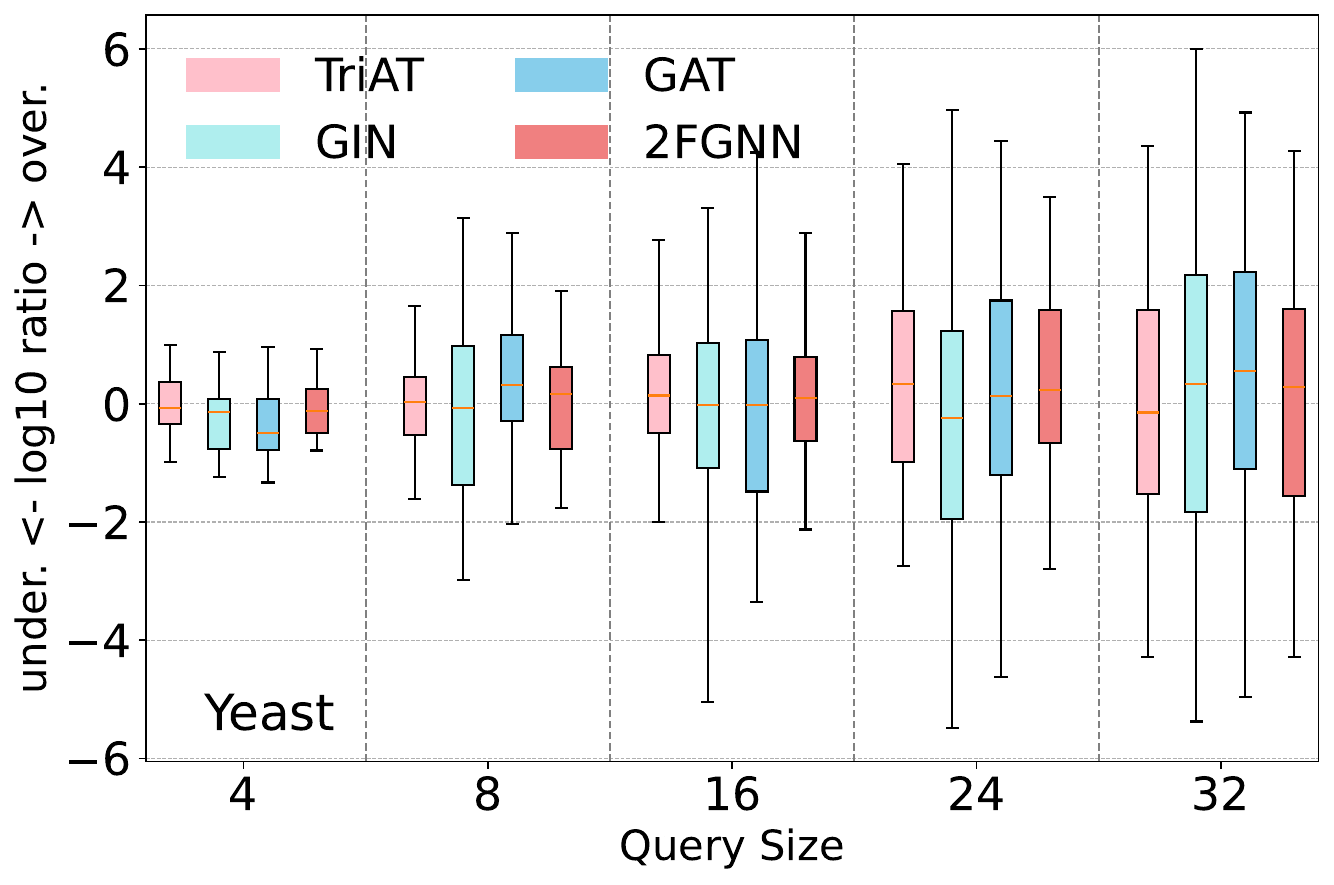}
        \captionsetup{skip=1.3pt}
        \caption{Yeast}
        \label{fig:yeast_ablation}
    \end{subfigure}
    \hfill
    \begin{subfigure}[c]{0.49\linewidth}
        \centering
        \includegraphics[width=\linewidth]{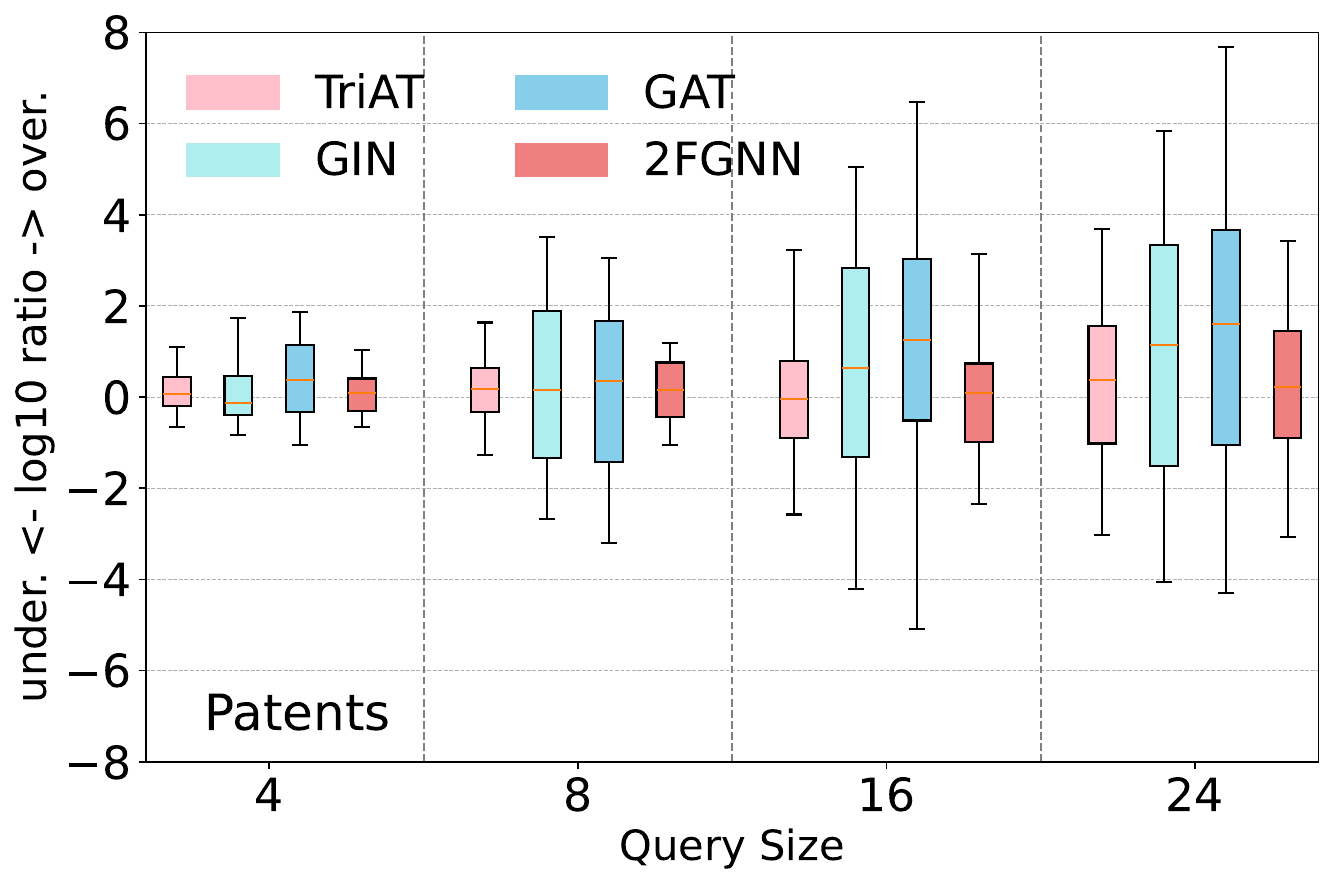}
        \captionsetup{skip=1.3pt}
        \caption{Patents}
        \label{fig:patents_ablation}
    \end{subfigure}
    \caption{Model ablation for cardinality estimation.}
    \label{fig:ablation_card}
\end{minipage}
\end{figure*}


NeurSC incurs the highest time cost due to its use of a GNN on the filtered graph. Fast likewise requires considerable time for filtering and sampling. In contrast, \method only gathers lightweight statistics during filtering, leading to improved efficiency. In our experiments, \method achieved a speedup of 1.75× to 222.12× over Fast across various query sizes and datasets.
LSS and GNCE apply a GNN directly to the query graph without a filtering stage, and are thus faster per invocation than \method.
However, the filtering in \method is performed only once and amortized across multiple model calls during subquery optimization.
As a result, when integrated into a query optimizer, \method benefits from amortized filtering costs and reduced redundant computation across subqueries, whereas LSS and GNCE are not designed for such use cases.

\subsubsection{Further Ablation Results for TriAT}
As an additional ablation study, we replace our proposed TriAT with other GNN architectures, including GIN \cite{xu2018how}, GAT \cite{velivckovic2018graph}, and 2FGNN \cite{maron2019provably}, for cardinality estimation task.
\figref{fig:ablation_card} presents the experimental results on the Yeast and Patents datasets.
As shown in the figure, TriAT consistently outperforms both GIN and GAT, which aligns with its more expressive power.
Moreover, while TriAT exhibits comparable accuracy with more expressive 2FGNN, it achieves this with lower computational cost (2FGNN requires up to 8.16 times more computation time than TriAT for query graphs with 32 vertices). This highlights TriAT’s favorable trade-off between accuracy and efficiency.

\subsection{Experiments on Edge-labeled Graphs}
\label{sec:kuzu_exp}

In addition to the datasets containing only vertex labels, we also conduct experiments on the Labeled Subgraph Query Benchmark (LSQB) \cite{mhedhbi2021lsqb}, a directed benchmark featuring both vertex and edge labels based on LDBC SNB \cite{erling2015ldbc}. LSQB focuses on complex join queries that are typical in social network analysis. Our experiments were conducted using LSQB with a scale factor of 3, which contains 11.3 million vertices and 66.2 million edges.

Since LSQB includes a schema graph representing possible vertex linkages, certain learning-based optimization methods for relational databases, such as JGMP \cite{reiner2023sample}, can also be applied. We therefore include it in our experiments for comparison.
As the original LSQB only provides 9 queries, we generate 75 queries for each number of query vertices from $\{3, 4, 5, 6, 7, 8\}$. We randomly split 80\% of the generated queries for training and use the remaining 20\% for testing. Reported results are based on the test set. 
\ifbool{fullversion}{
More details are provided in Appendix \ref{sec:appendix_more_kuzu_details}.}{\textcolor{red}{More details can be found in the full version of our paper \cite{}.}}

\begin{figure}
    \centering
    \begin{subfigure}[c]{0.45\linewidth}
        \centering
        \includegraphics[width=0.97\linewidth]{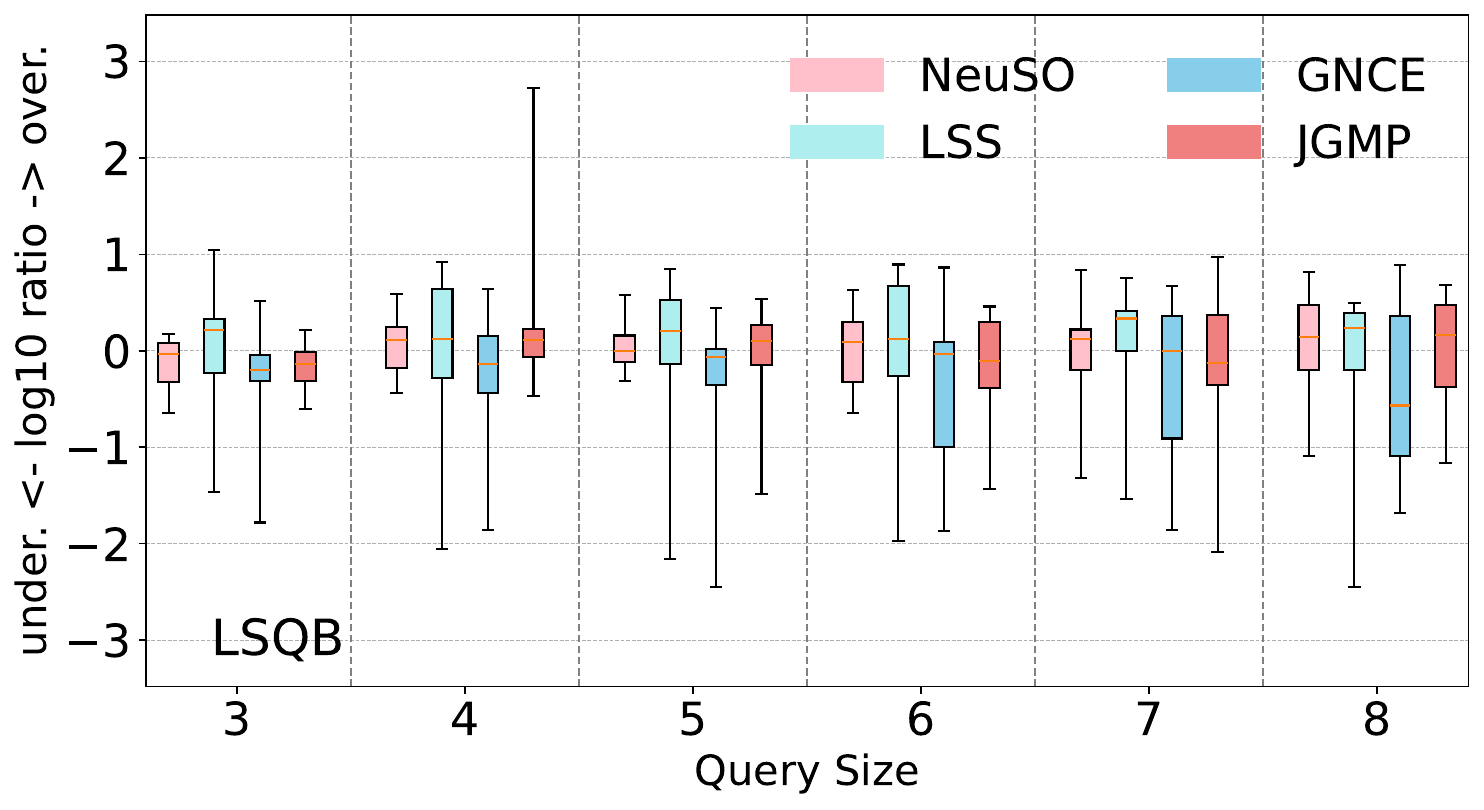}
        \captionsetup{skip=1.5pt}
        \caption{Cardinality estimation accuracy on LSQB.}
        \label{fig:lsqb_card}
    \end{subfigure}
    \hspace{0.03cm}
    \begin{subfigure}[c]{0.52\linewidth}
        \centering
        \includegraphics[width=0.78\linewidth]{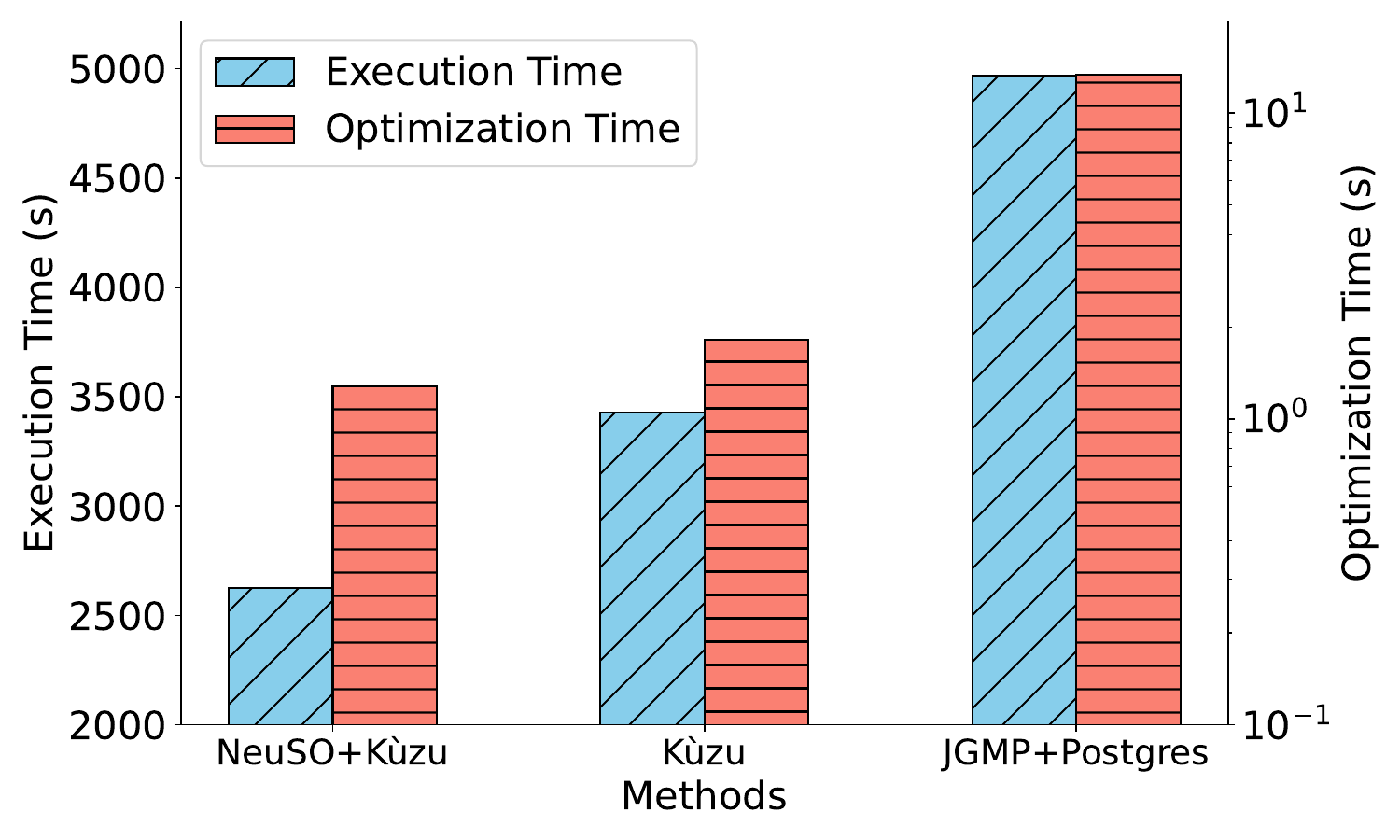}
        \captionsetup{skip=1.3pt}
        \caption{Execution \& optimization time comparison on LSQB.}
        \label{fig:lsqb_exe}
    \end{subfigure}
    \caption{Experiments on LSQB.}
    \label{fig:lsqb_exp}
\end{figure}

\subsubsection{Cardinality Estimation Accuracy}
We compare \method with other subgraph cardinality estimation methods (LSS and GNCE) that support directed and edge-labeled graphs. Methods such as NeurSC, Alley, and Fast are excluded as they do not handle directed edge-labeled graphs. We also include JGMP \cite{reiner2023sample}, a state-of-the-art cardinality estimator for relational databases, for comparison. As shown in \figref{fig:lsqb_card}, \method achieves a good balance between accuracy and stability, consistent with the results in \secref{sec:exp_ret_card}.

\subsubsection{Results for Optimization}
\label{sec:kuzu_ret_opt}
We further evaluate the optimization performance of \method on K{\`u}zu \cite{feng2023kuzu}, an open-source graph database designed for complex, join-heavy analytical workloads.
Since the baseline methods discussed in \secref{sec:exp_opt} do not support directed edge-labeled graphs, we exclude them from this experiment. Instead, we inject the query plans generated by \method into K{\`u}zu and compare them with plans produced by K{\`u}zu's built-in DP-based optimizer.
Additionally, we consider a relational baseline by replacing PostgreSQL’s cardinality estimator with JGMP \cite{reiner2023sample}, and executing the queries on PostgreSQL \cite{postgres}. A uniform timeout of 10 minutes is set for each query across all methods.

\figref{fig:lsqb_exe} presents the comparison in terms of total execution time (including optimization time). On K{\`u}zu, the execution plans generated by \method run 1.31$\times$ faster than those generated by the default optimizer, and the optimization process itself is 1.41$\times$ faster than K{\`u}zu’s optimizer. Due to differences in storage and execution engines, the relational method is significantly slower. These results demonstrate that \method can generate more efficient execution plans and is applicable in real-world systems.

\subsection{Robustness and Transferability}

In this section, we evaluate the robustness and transferability of \method under two realistic scenarios: 
(i) updates to the data graph, and (ii) workload shifts. 
We test on Yeast dataset as the queries are evenly distributed on various sizes and true cardinality.

\begin{figure}
    \setlength{\abovecaptionskip}{0.1cm}
    \centering
    \begin{subfigure}[c]{\linewidth}
        \centering
        \includegraphics[width=0.62\linewidth]{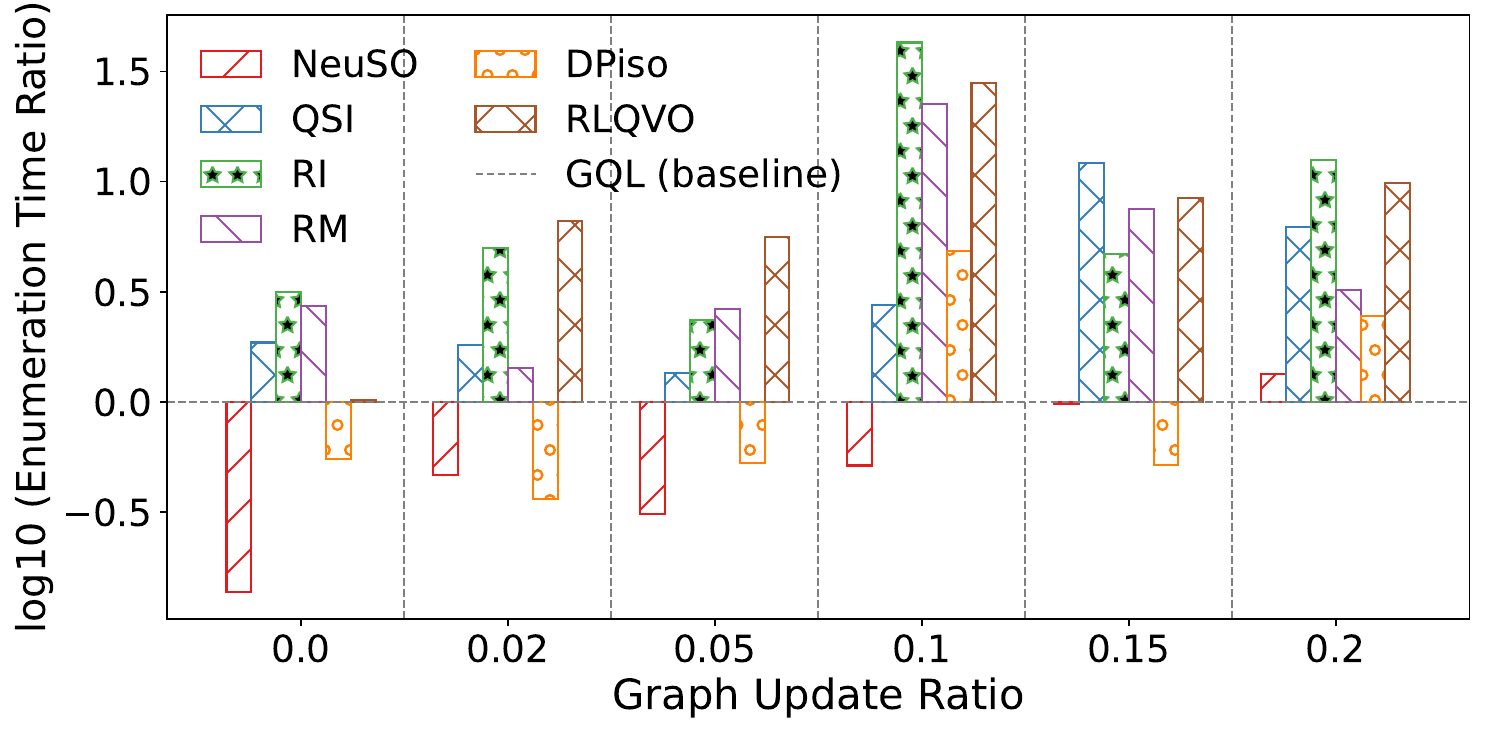}
        \captionsetup{skip=1.5pt}
        \caption{Enumeration time comparison under graph updates,  where GQL is taken as the baseline method.}
        \label{fig:yeast_enumeration_time_update}
    \end{subfigure}
    \begin{subfigure}[c]{\linewidth}
        \centering
        \includegraphics[width=0.82\linewidth]{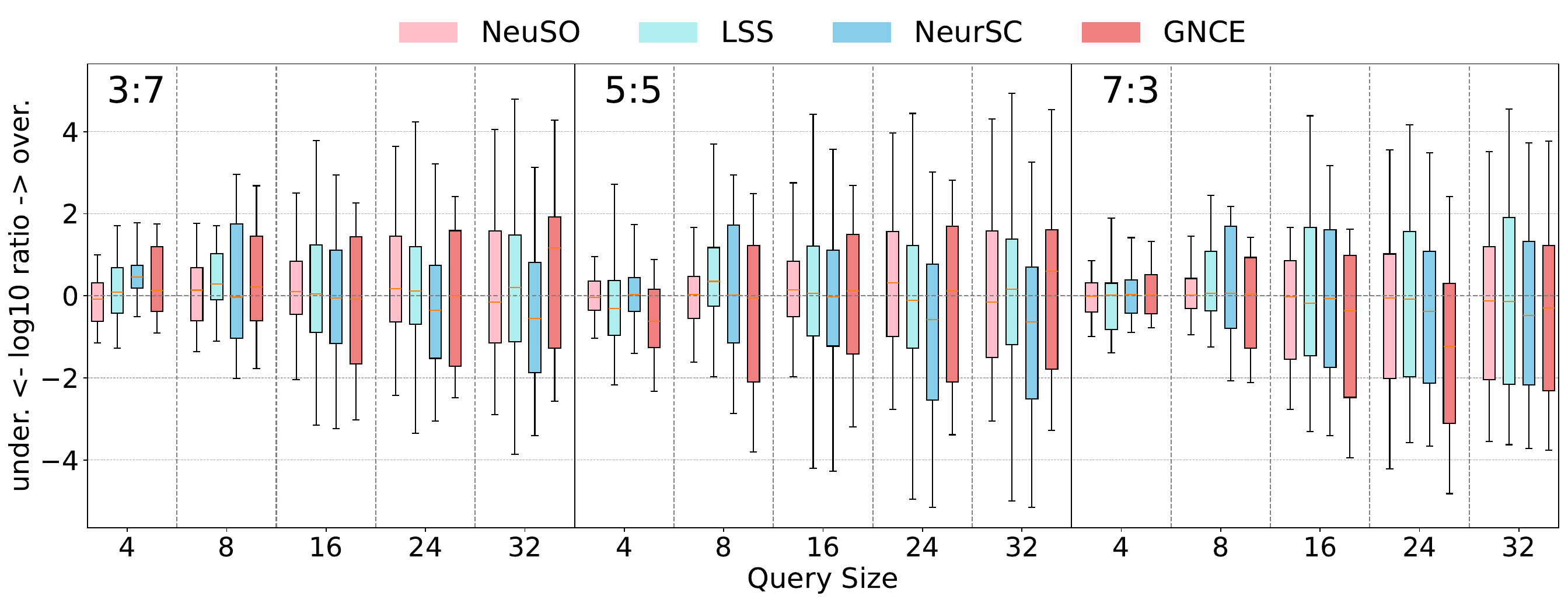}
        \captionsetup{skip=1.5pt}
        \caption{Cardinality estimation under varying training workload.}
        \label{fig:yeast_card_workload_shift}
    \end{subfigure}
    \caption{Robustness (a) and transferability (b) on Yeast.}
    \label{fig:transferability_robustness}
\end{figure}

\subsubsection{Robustness to Graph Updates}
Although \method is designed for static graphs, it can tolerate moderate changes in the data graph, as the query encoder receives input statistics that reflect the updated data graph after filtering.
We consider three types of updates: vertex label changes, edge insertions, and deletions, applied in a fixed ratio of $1{:}1{:}1$, with total update ratios of $0\%$, $2\%$, $5\%$, $10\%$, $15\%$, and $20\%$. 
As the data graph differs across update ratios, absolute enumeration time is incomparable. Thus, we use GQL as the baseline within each ratio to compute relative enumeration times.

The results are presented in \figref{fig:yeast_enumeration_time_update}. When the update ratio is no more than $10\%$, \method still produces high-quality matching orders. However, when the graph changes significantly (i.e., over $10\%$), \method becomes slower than GQL and DPiso, suggesting that retraining is necessary.

\subsubsection{Workload Transferability}
We evaluate the transferability of \method by constructing diverse workloads and comparing its estimation accuracy with other learning-based methods, including LSS, NeurSC, and GNCE.
Specifically, we enlarge the query pool to a size of 2000, consisting of 400 queries for each query size in $\{4, 8, 16, 24, 32\}$, whose cardinalities can be computed without exceeding limit time.
We designate 500 queries as the test set (100 for each size), and construct three distinct training workloads, each containing 1000 queries drawn from sizes 4, 8, 16, and 24. Queries of size 32 are excluded from training and used solely for evaluating generalization to unseen query sizes.
In each workload, the ratio between small queries (sizes 4 and 8) and large queries (sizes 16 and 24) varies among $\{3{:}7, 5{:}5, 7{:}3\}$.

The result is shown in \figref{fig:yeast_card_workload_shift}. Existing methods tend to underestimate the cardinalities of large queries in workloads dominated by small queries, and overestimate the cardinalities of small queries in workloads dominated by large queries. In contrast, \method achieves more stable estimation performance across different workload compositions, as it leverages information from subqueries of large queries during training.

\section{Related Works \& Discussions}

\subsection{Related Works}




\subsubsection{Learned Query Optimization for Relational Databases.}
Learning-based optimizers for relational databases can be categorized into three types:
(1) learning-based cardinality estimators combined with traditional cost estimators and DP plan enumerators \cite{kipf2019mscn,hilprecht13deepdb,yang2020neurocard,reiner2023sample,li2023alece}, (2) learning-based cost estimators with traditional DP plan enumerators \cite{marcus2019plan,  sun2019end,marcus2021bao}, and (3) direct learning-based query plan generators \cite{marcus2018deep,yu2020reinforcement,yang2022balsa,chen2023loger}. Each of these approaches replaces one or more components of the traditional optimizer to improve efficiency and accuracy.
There are also some recent works that utilize LLMs for optimization. Most of them focus on hint generation \cite{akioyamen2024unreasonable, yao2025query} or query rewriting \cite{liu2024query, zhou2021learned, li2024llm}.

\subsubsection{Sugraph Query Optimizer}

Subgraph matching order has been widely studied, and most approaches rely on relatively simple heuristics. For instance, GraphQL \cite{he2008graphs} processes queries by prioritizing vertices with smaller candidate sets before moving on to those with larger sets. Similarly, CFL \cite{bi2016efficient} and RapidMatch \cite{sun2020rapidmatch} initiate the matching process from the denser regions of the query graph. DPiso \cite{han2019efficient} and VEQ \cite{kim2021versatile} further enhance the matching process by dynamically adjusting the matching order during execution. RLQVO \cite{wang2022reinforcement} proposes a matching order method based on reinforcement learning, which models the matching order selection as an MDP. In addition to optimizations for general subgraph queries, several works focus on specific structures, e.g., paths \cite{yakovets2016query,nguyen2022estimating,pang2024materialized}.


In practical graph database systems, most modern systems (e.g., Neo4j \cite{neo4j}, GraphFlow \cite{mhedhbi2019optimizing}, K{\`u}zu \cite{feng2023kuzu}, and gStore \cite{yang2022gcbo}) employ dynamic programming to determine optimal execution plans. When dealing with complex query graphs, these systems may also resort to greedy heuristic search strategies. 
Some optimization methods \cite{DBLP:conf/edbt/Abul-BasherYGCC21,leeuwen2022avantgraph,kalumin2025optimizing} rely on cost models tailored to the specific execution operators they adopt (such as factorization). These designs are orthogonal to our method, which can be applied independently of such operator-specific assumptions.


\subsubsection{Subgraph Counting}

Subgraph counting aims to provide cardinality estimates for the entire given query graph.
Exact methods compute the count by exhaustively enumerating all subgraph matches in the data graph \cite{paredes2013towards,hovcevar2017combinatorial,himanshu2017impact}, while approximate methods utilize graph summaries or sampling strategies for estimation.
For instance, CharacteristicSets \cite{neumann2011characteristic} decomposes queries into star-shaped subqueries and estimates them via precomputed statistics; SumRDF \cite{stefanoni2018estimating} summarizes the graph by label and estimates counts on the abstracted structure.
Sampling-based approaches further improve efficiency by probabilistically exploring subgraphs \cite{vengerov2015join,li2016wander,zhao2018random,kim2021combining,yang2022gcbo,shin2024cardinality}. These methods focus on reducing bias and variance through advanced sampling strategies.

Recent advances leverage machine learning to model query-data graph interactions for count prediction \cite{liu2020neural,zhao2021learned,wang2022neural,schwabe2024cardinality}, offering a paradigm shift from traditional algorithmic designs.

\subsection{Discussions}

Compared with reinforcement learning-based optimizer RLQVO \cite{wang2022reinforcement} and other learning-based subgraph counting approaches including LSS \cite{zhao2021learned}, NeurSC \cite{wang2022neural}, and GNCE \cite{schwabe2024cardinality}, our \method demonstrates three key distinctions:

\textbf{Enhanced Graph Encoding.} While RLQVO adopts GCN \cite{kipf2017semisupervised} and both LSS and NeurSC employ GIN \cite{xu2018how}, more recent efforts such as GNCE \cite{schwabe2024cardinality} propose advanced encoders like TPN. In comparison, \method introduces TriAT, which provides stronger expressive power than these message-passing neural networks.

\textbf{Multi-Task Capability.} Existing methods exhibit functional limitations: RLQVO focuses solely on matching order generation, while LSS, NeurSC, and GNCE estimate only the full query graph's cardinality. In contrast, \method establishes a unified multi-task framework that simultaneously predicts both cardinality and computational cost for arbitrary subqueries, while maintaining the capability to generate high-performance matching orders.

\textbf{Efficient Cross-Graph Interaction.} 
All of RLQVO, LSS, and GNCE rely on static statistical features from data graphs (loose statistics on the number of vertices (RLQVO) or pre-learned embeddings on the data graph (LSS and GNCE)). Additionally, RLQVO needs to perform GNN computations iterately during matching order selection. 
Although NeurSC incorporates dynamic graph filtering, its dual GNN architecture (query graph + filtered data graph) incurs significant overhead. 
\method innovates through a lightweight interaction mechanism that propagates filtered statistics through a single GNN operating exclusively on the query graph, which is more efficient than NeurSC while maintaining high accuracy.

\section{Conclusions}
\label{sec:conclusions}

In this paper, we propose \method, a unified and efficient framework for subgraph query optimization that generates high-quality vertex matching orders. 
\method employs a novel GNN-based architecture on the query graph and leverages a multi-task learning strategy to jointly predict subquery cardinalities and costs. 
Based on these predictions, we design a new top-down plan enumerator that selects the next vertex accordingly.
Experimental results show that \method outperforms existing methods in both enumeration time and end-to-end execution time.



\begin{acks}
This work was supported by The National Key Research and Development Program of China under grant 2023YFB4502303 and NSFC under grant 62532001. Lei Zou is the corresponding author of this work.
\end{acks}

\clearpage

\bibliographystyle{ACM-Reference-Format}
\bibliography{sample-base}

\ifbool{fullversion}{
    \appendix
    \clearpage

\section{Proof of Theorems}
\label{sec:proof_of_triat}

In this appendix, we provide the proofs of \thmref{the:triat} and \thmref{the:triat_complexity}, which establish the expressive power and computational complexity of TriAT.

\subsection{Proof of \thmref{the:triat}}

We will prove the following extension version of \thmref{the:triat}:

\begin{theorem}[TriAT's Expressive Power (Extension Version)]
The following expressiveness inclusion relation holds: 1-WL test = MPNN $\prec$ TriAT $\prec$ 2-FWL.
\end{theorem}

We separate our proof into four parts. The detailed introduction to 2-FWL is omitted in this paper. Readers interested in 2-FWL can refer to \cite{cai1992optimal, huang2021short} for more information.

\nosection{1. 1-WL $\preceq$ TriAT.}

The TriAT model considers more comprehensive neighborhood linkages compared to the 1-WL test, making this expressiveness relation trivial. The TriAT will degrade into MPNNs with expressiveness equivalent to the 1-WL test if the neighborhood linkage aggregation $\tau^{(k)}$ in \eqnref{eqn:triat_framework} is removed or set to a zero function.

\nosection{2. 1-WL $\neq$ TriAT.}

To prove that two graph models have different expressiveness power, we only need to find two graphs where one model produces the same representations while the other produces different representations. The two graphs in \figref{fig:wl_test} serve as the touchstone. The 1-WL cannot discriminate between the two graphs, while TriAT can.

\nosection{3. TriAT $\preceq$ 2-FWL.}

We can consider the representation of a vertex $u$ in the graph as the set of its representations with all vertexes: $(u)\triangleq \{(u, v)| v\in V\}$. This representation will be updated to $\{\{((u, w), (w, v), (u, v))|w\in V\}| v\in V\}$ according to the update procedure of 2-FWL.

However, the representation of $u$ will be updated to $\{u, \{(u, v)|v\in N(u)\}, \{(w, v)|w, v\in N(u)\}\}$ in TriAT, which is weaker than $\{\{((u,\\ w), (u, v), (w, v))|w\in V\}| v\in N(u)\}$. From this, we can see that it is weaker than the update of 2-FWL.

\nosection{4. TriAT $\neq$ 2-FWL.}

Although TriAT recognizes the triangle patterns (a triangle is a cycle with three vertices) in graphs, it cannot recognize larger cycles in graphs. TriAT cannot discriminate between the two graphs in \figref{fig:2fwl}, while 2-FWL can. \figref{fig:2fwl_1} presents an unconnected graph with two 4-vertex cycles, and \figref{fig:2fwl_2} presents a graph with 8-vertex cycles.

\begin{figure}[htbp]
    \centering
    \begin{subfigure}[b]{0.3\linewidth}
        \centering
        \includegraphics[width=\linewidth]{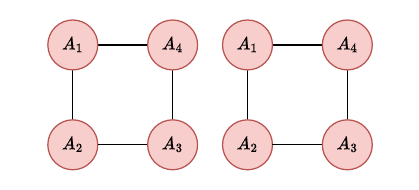}
        \caption{Graph 1}
        \label{fig:2fwl_1}
    \end{subfigure}
    \hspace{2.5cm}
    \begin{subfigure}[b]{0.3\linewidth}
        \centering
        \includegraphics[width=\linewidth]{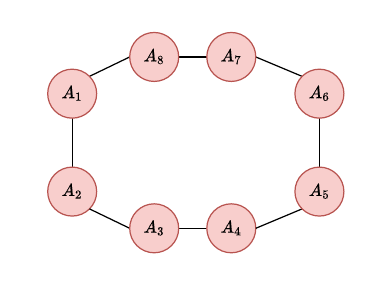}
        \caption{Graph 2}
        \label{fig:2fwl_2}
    \end{subfigure}
    \caption{Two graphs cannot be distinguished by TriAT. The vertexes with the same notations will get the same representation by TriAT.}
    \label{fig:2fwl}
\end{figure}

\subsection{Proof of \thmref{the:triat_complexity}}

Each layer of TriAT requires gathering information from neighbors ($\phi$) and their mutual connections ($\tau$) for every vertex.
The total computational complexity for $\phi$ and $\tau$ across all vertices is $O(|E|)$ and $O(|\Delta|)$, respectively.
This is because each edge is involved in $\phi$ through its two endpoints, and contributes to $\tau$ only if it is part of a triangle.
These two components together yield the overall time complexity of $O(|E| + |\Delta|)$. The space complexity is $O(|E|)$, as TriAT adopts a vertex-update framework that requires storing features for all vertices and edges.

\section{Equivalence Between the Shortest Path and the Optimal Join Plan}

\begin{theorem}
    The shortest path from $\emptyset$ to $Q$ in the CCG corresponds to the join order that minimizes the execution cost.
\end{theorem}

\begin{proof}
    Each join order corresponds to a distinct path from $\emptyset$ to $Q$ in the CCG. For any given path $P=(q_0\rightarrow q_1\rightarrow\cdots\rightarrow q_n)$, its total execution cost can be decomposed into several one-step costs: $cost(P) = \sum_{i=1}^n \mathcal{C}_Q(e(q_{i-1}, q_i))=l(P)$. Consequently, minimizing execution cost is equivalent to finding the path $P^*$ with minimal length $l(P^*)$ in $CCG_Q$.
\end{proof}

\section{Pseudocode of the Training Process (\secref{sec:training_process})}
\label{sec:alg_of_training_process}

To complement the description in \secref{sec:training_process}, we provide the pseudocode of \method's training process in \algref{alg:training_process}.

{\color{red}
\begin{small}
\begin{algorithm}[htbp]
\DontPrintSemicolon
\caption{Training process for one epoch}
\label{alg:training_process}
\KwIn{The training set $\mathcal{T}$}
\ForEach{$Q\in \mathcal{T}$}{
    $\bm{X}^{(0)}\leftarrow$ Initialize vertices' features \tcp*[r]{\secref{sec:feature_init}}
    $\bm{X} \leftarrow TriAT(Q, \bm{X}^{(0)})$\tcp*[r]{\secref{sec:triangle_gnn}} 
    $loss\leftarrow 0$\;
    \If {$Q.state = Full$}{
        \For{$i \in \{card, cost, mc\}$}{
            \ForEach(\tcp*[f]{for each batch}){$bt_i \in Q.i$} {
                $loss\leftarrow loss+\mathcal{L}_q(MLP_{i}(X_{bt_i}), y_{bt_i})\times \alpha_{i}$\;
            }
        }
    }
    \Else(\tcp*[f]{$Q.state = Partial$}) { 
        \For{$i \in \{card, cost\}$}{
            \ForEach{$bt_i \in Q.i$} {
                $loss\leftarrow loss+\mathcal{L}_q(MLP_{i}(X_{bt_i}), y_{bt_i})\times \alpha_{i}$\;
            }
        }
    }
    $loss \leftarrow loss + \sum_{q_0\in Q_e}\mathcal{L}_c(q_0)$\tcp*[r]{constraint loss}
    Update models' parameters with gradient descent on $loss$\;
}


\end{algorithm}
\end{small}
}

\section{More Experiment Results}
\label{sec:appendix_more_exp_results}
This section presents additional results on optimization and cardinality estimation that were omitted from the main text due to space limitations.

\figref{fig:enumeration_time_comp} presents the enumeration time comparison on four datasets that were not included in \secref{sec:exp_enumeration_time_comparison}.

\begin{figure}
    \centering
    \begin{subfigure}[c]{0.49\linewidth}
        \centering
        \includegraphics[width=0.9\linewidth]{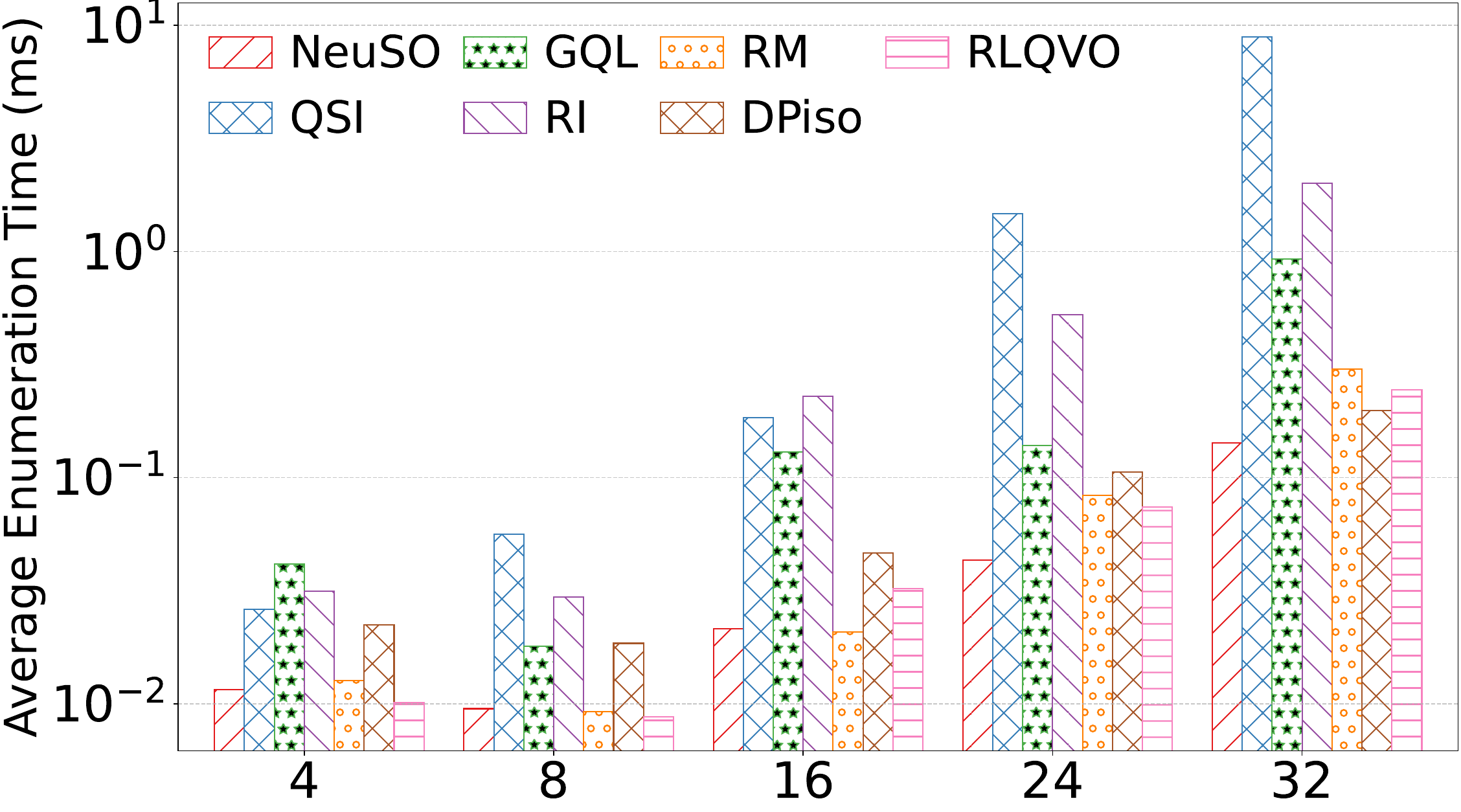}
        \caption{HPRD}
        \label{fig:hprd_enumeration}
    \end{subfigure}
    \begin{subfigure}[c]{0.49\linewidth}
        \centering
        \includegraphics[width=0.9\linewidth]{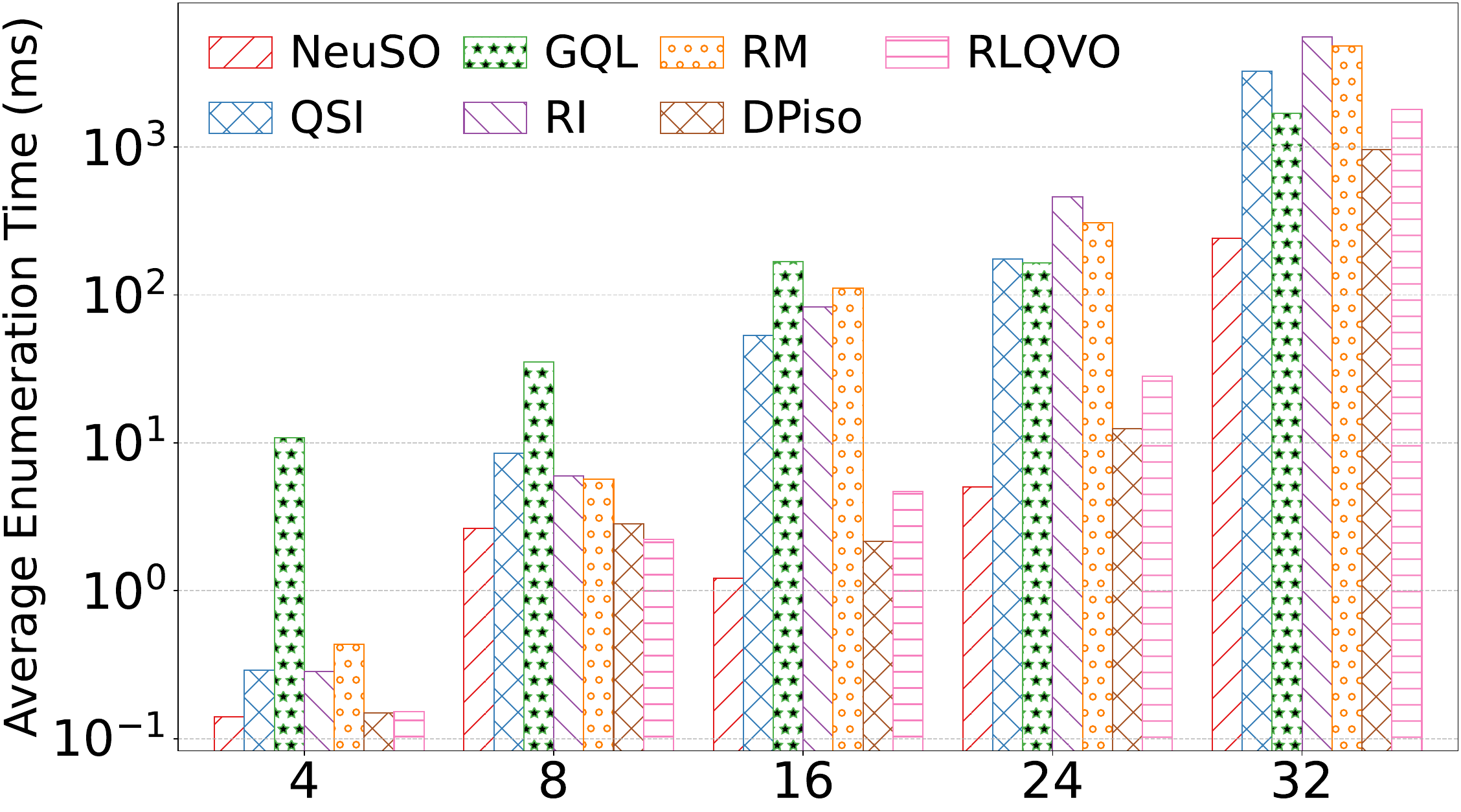}
        \caption{Yeast}
        \label{fig:yeast_enumeration}
    \end{subfigure}
    \begin{subfigure}[c]{0.49\linewidth}
        \centering
        \includegraphics[width=0.9\linewidth]{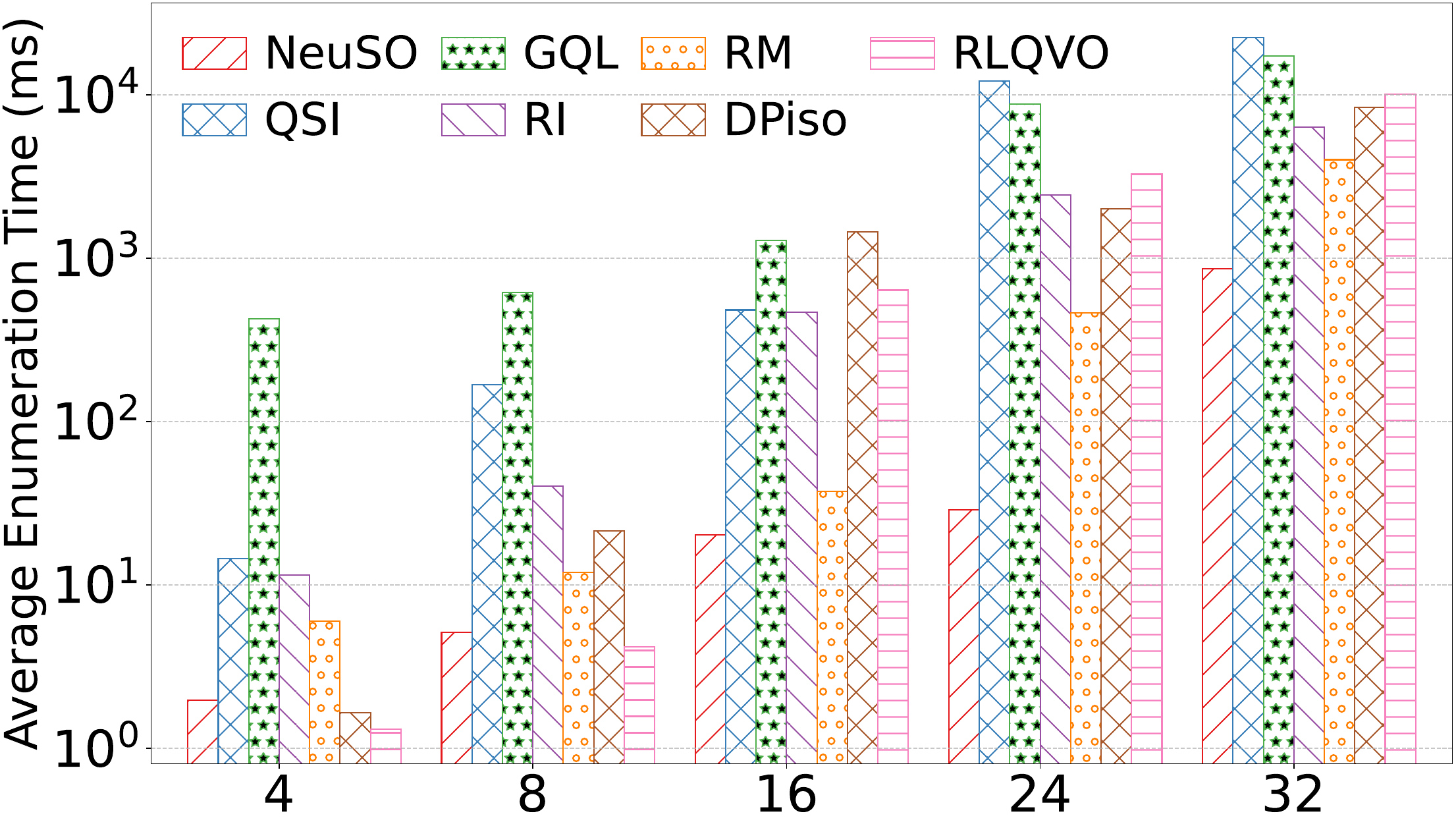}
        \caption{YouTube}
        \label{fig:youtube_enumeration}
    \end{subfigure}
    \begin{subfigure}[c]{0.49\linewidth}
        \centering
        \includegraphics[width=0.9\linewidth]{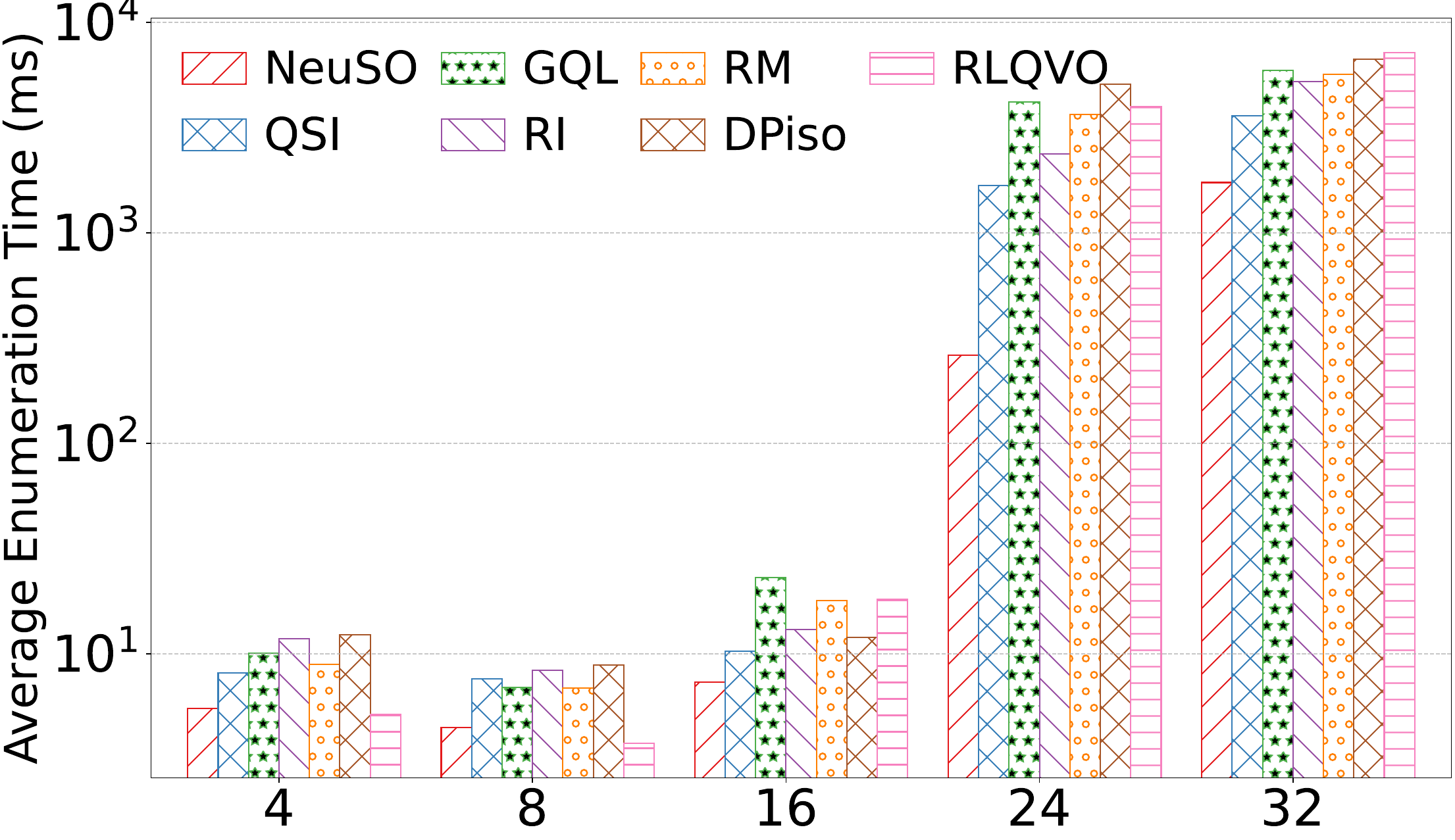}
        \caption{Patents}
        \label{fig:patents_enumeration}
    \end{subfigure}
    \caption{More average enumeration time comparison results. The x-axis represents different query graph sizes $V_Q$.}
    \label{fig:enumeration_time_comp}
\end{figure}

\figref{fig:dblp_eu2005_end_to_end} shows the end-to-end running time comparison on DBLP and EU2005, complementing the results in \secref{sec:end_to_end_time_comparison}.

\begin{figure}
    \centering
    \begin{subfigure}[c]{0.49\linewidth}
        \centering
        \includegraphics[width=0.9\linewidth]{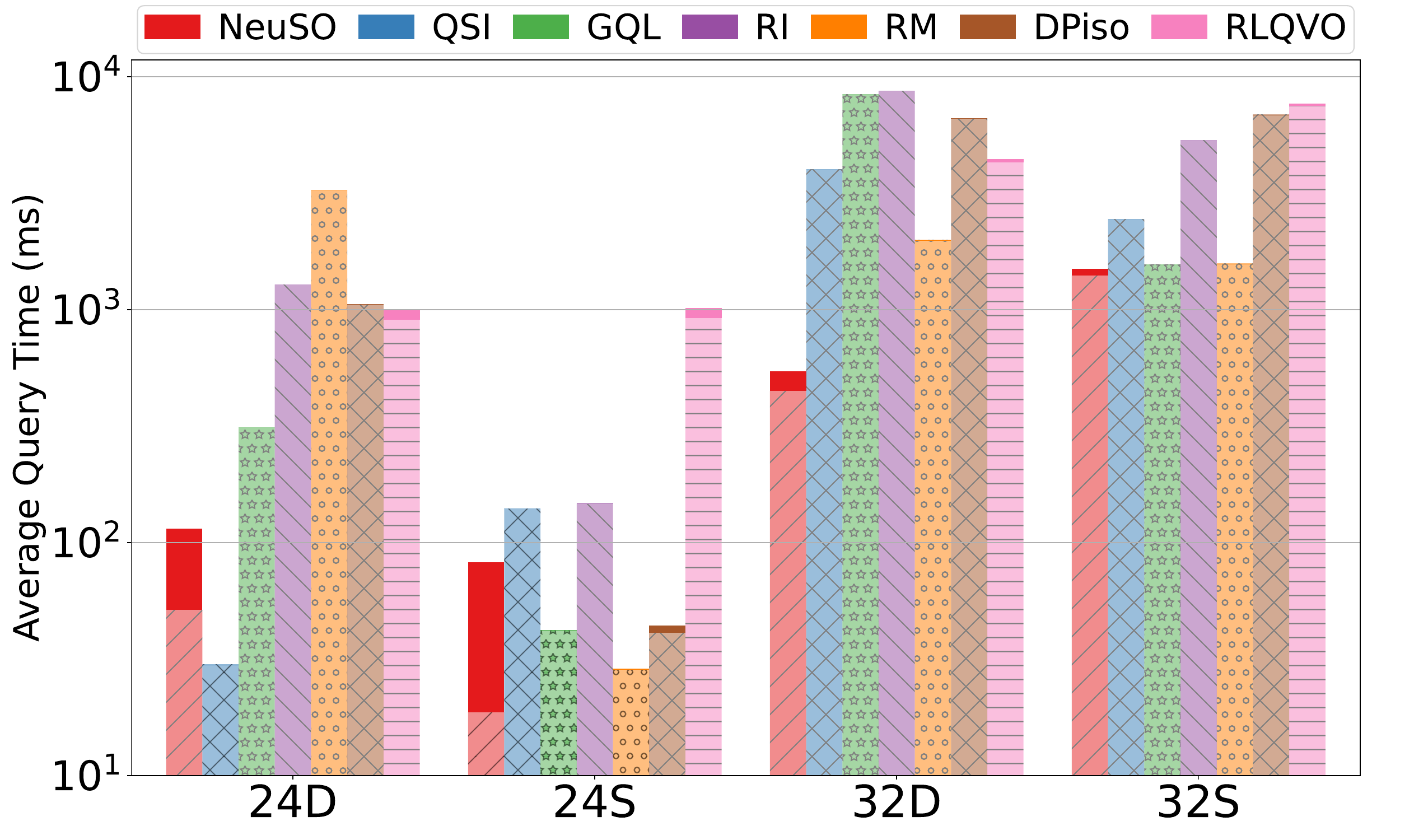}
        \caption{DBLP}
        \label{fig:dblp_end_to_end}
    \end{subfigure}
    \begin{subfigure}[c]{0.49\linewidth}
        \centering
        \includegraphics[width=0.9\linewidth]{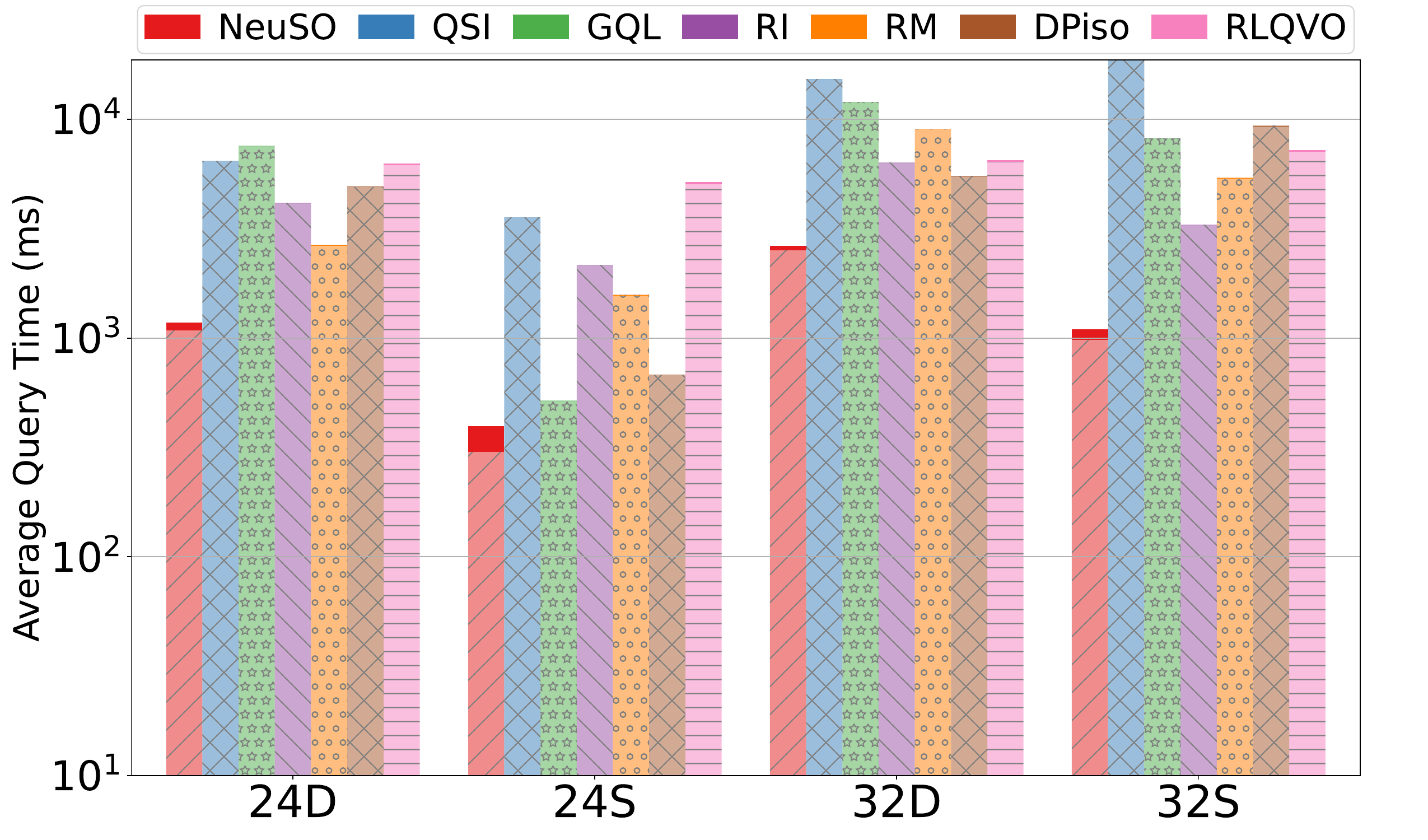}
        \caption{EU2005}
        \label{fig:eu2005_end_to_end}
    \end{subfigure}
    \caption{Average end-to-end query time comparison. The bottom bars with hatching represent filter and enumeration time, while the solid bars above represent optimization times.}
    \label{fig:dblp_eu2005_end_to_end}
\end{figure}

\figref{fig:card_est_time_comp_more} presents the cardinality estimation time comparison on YouTube and Patents, as a continuation of the analysis in \secref{sec:card_efficiency_exp}.

\begin{figure}
    \centering
    \begin{subfigure}[c]{0.49\linewidth}
        \centering
        \includegraphics[width=0.9\linewidth]{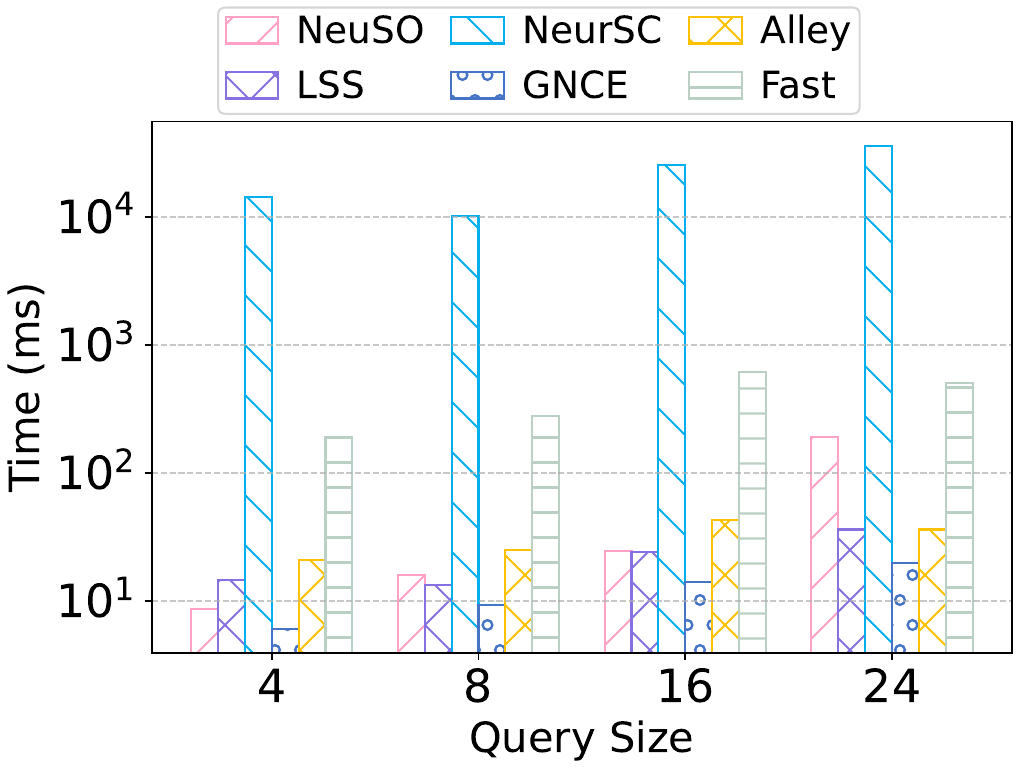}
        \caption{YouTube}
        \label{fig:card_time_youtube}
    \end{subfigure}
    \begin{subfigure}[c]{0.49\linewidth}
        \centering
        \includegraphics[width=0.9\linewidth]{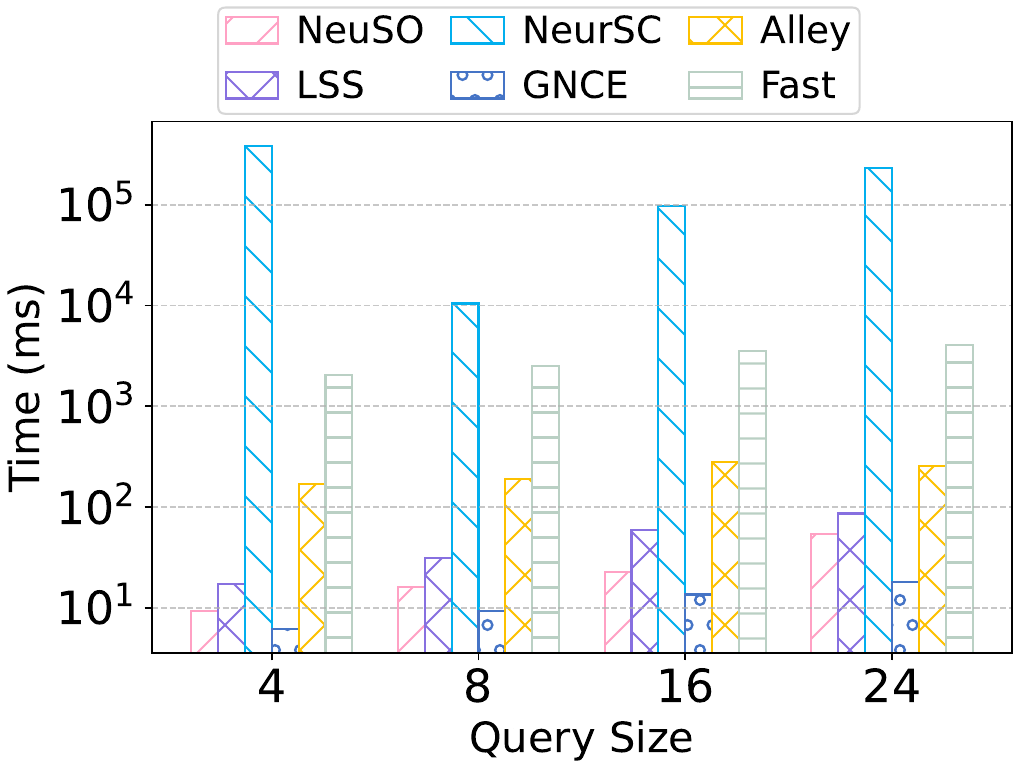}
        \caption{Patents}
        \label{fig:card_time_patents}
    \end{subfigure}
    \caption{Average cardinality estimate time on YouTube and Patents.}
    \label{fig:card_est_time_comp_more}
\end{figure}

\section{Additional Experiment Details for \secref{sec:kuzu_exp}}
\label{sec:appendix_more_kuzu_details}

We used the LSQB benchmark \cite{mhedhbi2021lsqb} with a scale factor of 3 in our experiments, which contains 11.3 million vertices and 66.2 million edges.
The original query set provided by LSQB includes queries with \texttt{Optional} and \texttt{Negative} edges. To ensure compatibility and completeness, we also generate queries containing these two types of edges. Specifically, after retrieving a query skeleton from the schema graph, we randomly assign some edges to be optional (with a probability of 0.15) or negative (with a probability of 0.1). We ensure that these two types of edges do not coexist in a single query, and that there is at most one negative edge per query. This process increases the diversity of the generated queries. We ensure that all generated queries are distinct and yield valid results.

To support these special edge types (optional and negative), we concatenate a one-hot encoding to each edge in the query, indicating its type (normal, optional, or negative). The same encoding scheme is applied to all compared methods to ensure a fair comparison.

We use Kùzu version 0.10.0 in our experiments, and we assign the candidate size of each vertex and edge (i.e., $|C(u)|$ in \eqnref{eqn:vertex_feature} and $|C\langle u_1, u_2\rangle|$ in \eqnref{eqn:dir_labeled_edge_feature}) to be the size of the corresponding vertex and edge tables, respectively.

In our experiment of \secref{sec:kuzu_ret_opt}, we set the number of working threads to 1 for both K{\`u}zu and PostgreSQL. \figref{fig:detailed_lsqb_size} provides a more detailed comparison of execution times. It can be observed that in most cases, \method produces more effective execution plans, resulting in a speedup of up to 42.4×. This improvement is mainly due to \method's more accurate estimation of execution costs.

\nop{
However, there are a query of size 4, whose plan generated by \method is slower than that of K{\`u}zu. This is because K{\`u}zu applies specific optimizations for such queries, avoiding the need to materialize all intermediate results—an optimization that cannot be leveraged by the injected plans.
}

\begin{figure}
    \centering
    \includegraphics[width=0.7\linewidth]{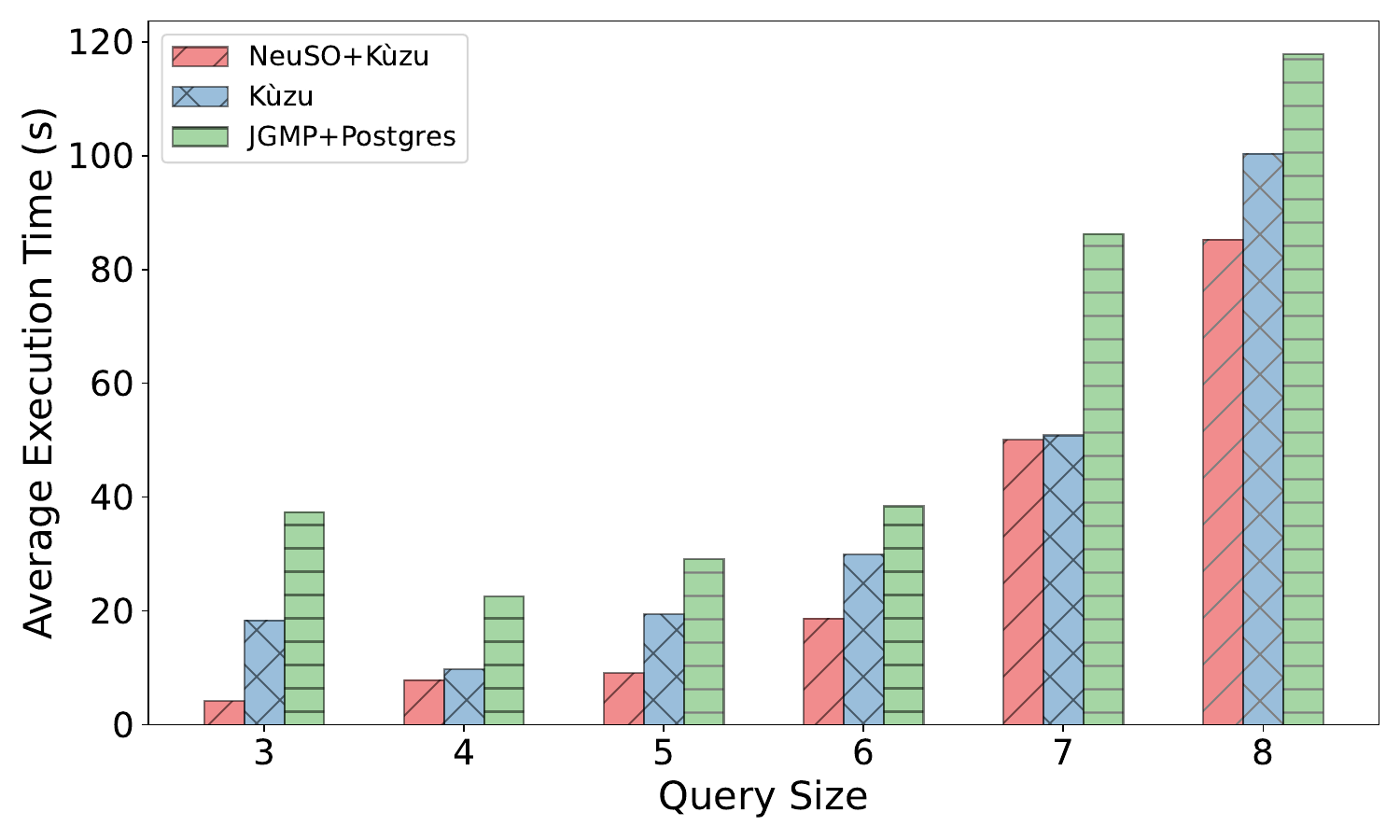}
    \caption{Detail execution time comparison on LSQB.}
    \label{fig:detailed_lsqb_size}
\end{figure}

\section{More Discussions}

\subsection{Rationality Discussion about GNN After Filtering}

During the running process of TriAT, a vertex $u$ perceives its $k$-hop neighborhood after $k$ rounds of message-passing. For the entire query graph $q$, the pooling methods introduced above will be no problem. However, for the represention of a subquery $q_{sub}$, it would inevitably receives more information than $q_{sub}$ itself. 
Specifically, $\bm{x}_{q_{sub}}$ may contains information that extends beyond its own structure. For example, after one layer of TriAT, the representation of subquery $q_{\{u_2, u_3, u_4\}}$ in \figref{fig:triat_example} will include information from vertex $u_5$, which is not part of $q_{\{u_2, u_3, u_4\}}$.

This apparent inconsistency is intentional and serves two key purposes.
First, during the filter stage, candidates for vertex $u$ are determined based on the candidate sets of vertices that are beyond $u$'s immediate neighbors. This means that the subquery would be influenced by query vertices that do not belong to it.
Second, this overlap of information helps to capture the relationships between subqueries, which ultimately aids in query optimization by capturing inter-subquery dependencies.

\subsection{Discussions about the Top-down Plan Enumerator.}  

The plan enumeration approach employed by \method follows a top-down paradigm, distinguishing it from reinforcement learning-based optimization techniques such as RLQVO, which rely on bottom-up action selection. The advantages of \method's methodology can be summarized as follows:  

\begin{itemize}[leftmargin=8pt]  
    \item \textbf{History Information Reusability.}  
    The top-down enumeration strategy enables the reuse of historical optimization results. For instance, if a subquery matches a previously processed query in the log, the system might leverage cached vertex ordering information to eliminate redundant computations, though this optimization is not implemented in our current work. 

    \item \textbf{State Cost as an Intrinsic Property.}  
    In \method, the minimum cost associated with a state (subquery) is defined as an intrinsic property of the state itself (or more precisely, of the state and the initial empty state, which remains consistent across all queries). In contrast, reinforcement learning methods like RLQVO evaluate states based on their relation to the final state (i.e., the complete query $Q$), which varies per query. Consequently, for states of identical dimensionality, \method's cost computation is more efficient and less computationally demanding, which also reduces the learning difficulty. 
\end{itemize}

}

\end{document}